%% file: main.tex
\documentclass[journal,twoside,web]{ieeecolor}

\input{preamble/additional-packages}

\input{preamble/custom-definitions}
\input{preamble/acronyms}

\def\BibTeX{{\rm B\kern-.05em{\sc i\kern-.025em b}\kern-.08em
    T\kern-.1667em\lower.7ex\hbox{E}\kern-.125emX}}
\definecolor{abstractbg}{rgb}{0.89804,0.94510,0.83137}
\begin{document}
	
\clearpage
\onecolumn
\thispagestyle{empty}

\vspace*{\fill}

\begin{center}
	\large \textcopyright~2026 IEEE. Personal use of this material is permitted.
	Permission from IEEE must be obtained for all other uses, in any current or
	future media, including reprinting/republishing this material for advertising
	or promotional purposes, creating new collective works, for resale or
	redistribution to servers or lists, or reuse of any copyrighted component of
	this work in other works.
\end{center}

\vspace*{\fill}

\clearpage
\twocolumn

\title{Scheduling Mechanisms in Wireless Sensor-Actuator Networks for Multi-rate Periodic Control in Industry 4.0}
\author{Dingwen Yuan, Luis F. Abanto-Leon, and Matthias Hollick
\thanks{This work has been supported in part by the German Science Foundation (DFG) within the project CRUST (Grant 503199853) and by the LOEWE initiative (Hesse, Germany) within the emergenCITY center [LOEWE/1/12/519/03/05.001(0016)/72].}
}

\IEEEtitleabstractindextext{%
\fcolorbox{abstractbg}{abstractbg}{%
\begin{minipage}{\textwidth}%
\begin{wrapfigure}[13]{r}{3in}%
	\includegraphics[width=2.85in]{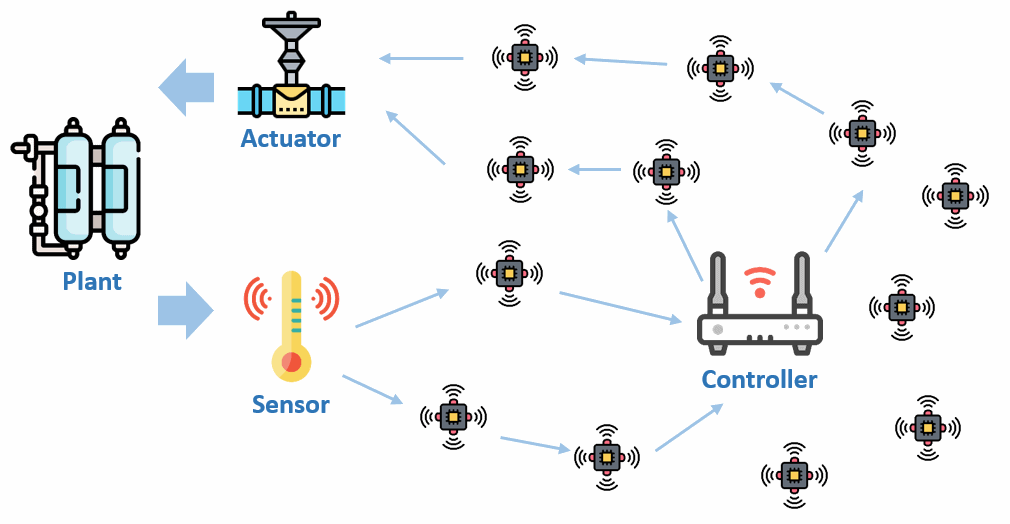}%
\end{wrapfigure}%

\input{preamble/abstract}
\input{preamble/keywords}
\end{minipage}}}

\maketitle

\input{sections/introduction}
\input{sections/related-work}

\input{sections/system-model}

\input{sections/proposed-framework}

\input{sections/proposed-method}

\input{sections/simulation-results}
\input{sections/extension-large-network}
\input{sections/conclusions}

\input{sections/appendices}

\bibliographystyle{IEEEtran}
\bibliography{ref}

\begin{IEEEbiography}[{\includegraphics[width=1in,height=1.25in,clip,keepaspectratio]{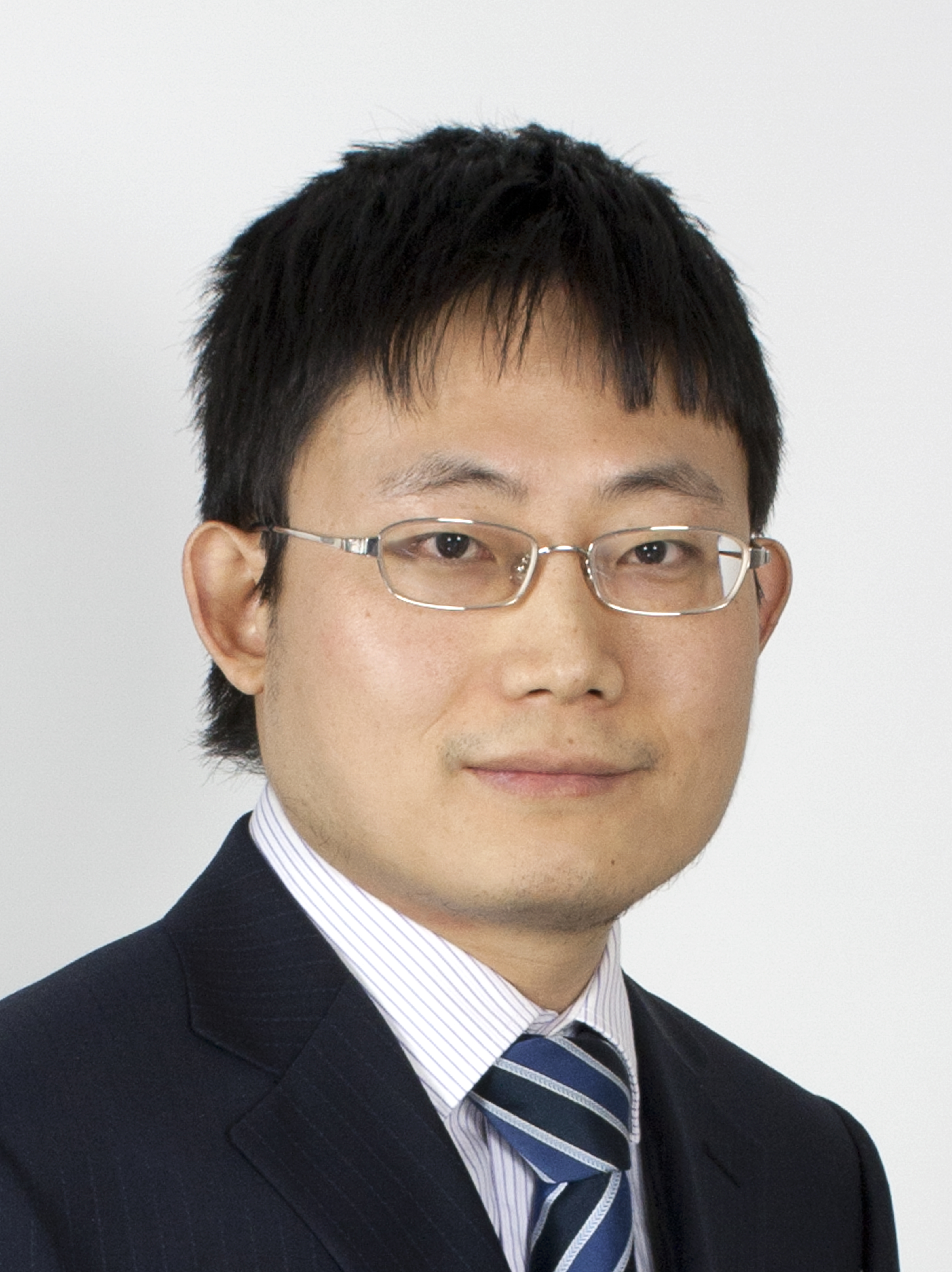}}]{Dingwen Yuan} He received his Bachelor of Engineering degree from the Department of Automation, Shanghai Jiao Tong University in 2002. In 2010 and 2016, he got his Master of Science and Ph.D. degrees from Technische Universität Darmstadt. Currently he works as an engineer at Nokia. His research interests include the QoS guarantee in wireless sensor networks (WSN), and the scheduling and optimization of millimeter-wave networks and cellular networks in general.
\end{IEEEbiography}

\begin{IEEEbiography}[{\includegraphics[width=1in,height=1.25in,clip,keepaspectratio]{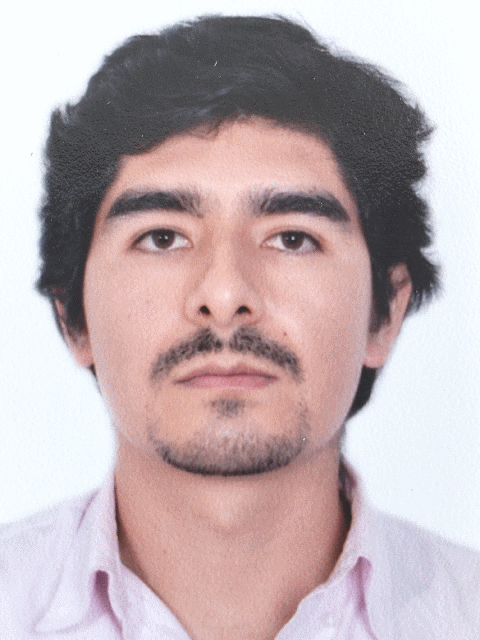}}]{Luis F. Abanto-Leon} earned his master's degree in Communications Engineering from Tohoku University, Japan, in 2015 and completed his Ph.D. in Computer Science at Technische Universität Darmstadt, Germany, in 2023. He is currently a postdoctoral researcher at Ruhr-Universität Bochum, Germany. His research interests include optimization theory, signal processing, and algorithm design for radio resource management in wireless networks.
\end{IEEEbiography}

\begin{IEEEbiography}[{\includegraphics[width=1in,height=1.25in,clip,keepaspectratio]{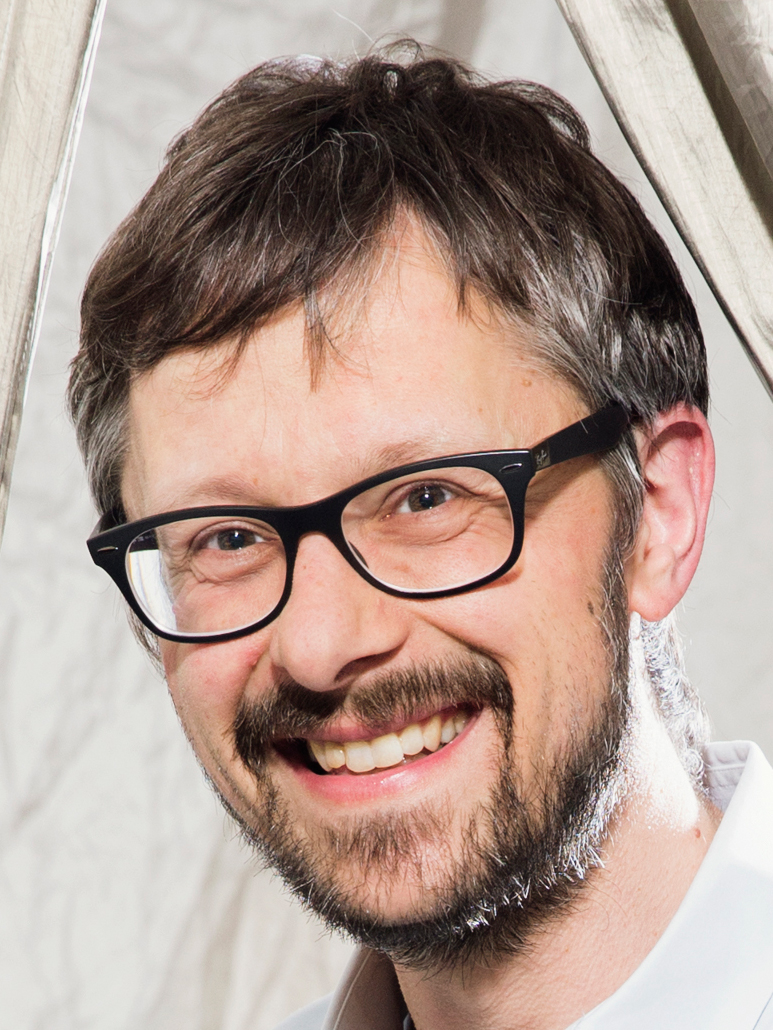}}]{Matthias Hollick} received the Ph.D. degree from Technische Universität (TU) Darmstadt in 2004. He is currently the Head of the Secure Mobile Networking Lab, Department of Computer Science, TU Darmstadt, Germany. He has been researching and teaching at TU Darmstadt, Universidad Carlos III de Madrid, and the University of Illinois at Urbana–Champaign. His research focus is on resilient, secure, privacy-preserving, and quality-of-service-aware communication for mobile and wireless systems and networks.
\end{IEEEbiography}

\end{document}

%% file: preamble/additional-packages.tex
\usepackage{jsen}
\usepackage{amsmath,amssymb,amsfonts}
\usepackage{cite}
\usepackage[linesnumbered,ruled,vlined]{algorithm2e} 
\usepackage{rtsched}
\usepackage{tikz}
\usetikzlibrary{arrows.meta,calc,decorations.markings,math,arrows.meta}
\usepackage{pifont}
\let\labelindent\relax
\usepackage{enumitem}
\usepackage{booktabs}
\usepackage[outdir=./]{epstopdf}
\usepackage[export]{adjustbox}
\usepackage{multirow}
\usepackage{graphicx}
\usepackage{pgfplots}
\usepackage{subcaption}
\usepackage{glossaries}
\usepackage{array}
\usepackage{comment}
\usepackage{multirow}
\usepackage{makecell}
\usepackage{graphicx}
\usepackage{tcolorbox}
\usepackage{textcomp}
\usetikzlibrary{shapes}
\usepackage{float}
\usepackage{dblfloatfix }
\usepackage[capitalise]{cleveref}
\usepackage{textcomp}
\usepackage{wrapfig}

\usepackage{microtype}

%% file: preamble/custom-definitions.tex
\newtheorem{thm}{Theorem}

\makeatletter
\newcommand{\algorithmfootnote}[2][\footnotesize]{%
  \let\old@algocf@finish\@algocf@finish
  \def\@algocf@finish{\old@algocf@finish
    \leavevmode\rlap{\begin{minipage}{\linewidth}
    #1#2
    \end{minipage}}%
  }%
}
\makeatother

\DeclareMathOperator{\lcm}{lcm}

\tikzstyle{sensor}=[circle, draw, thin,fill=cyan!20, scale=0.8]
\tikzstyle{actuator}=[rectangle, draw, thin,fill=green!20, scale=0.8]
\tikzstyle{controller}=[regular polygon,regular polygon sides=6, draw, thin,fill=orange!20, scale=0.8]
\tikzstyle{relay}=[regular polygon,regular polygon sides=3, draw, thin,fill=yellow!20, scale=0.4]

\SetKwInput{KwInit}{Initialize}
\SetAlFnt{\small}
\SetAlCapFnt{\small}
\SetAlCapNameFnt{\small}
\SetAlCapHSkip{1pt}
\IncMargin{\parindent}

%% file: preamble/acronyms.tex
\makeglossaries
\newacronym{TDMA}{TDMA}{time-division multiple access}
\newacronym{WSAN}{WSAN}{wireless sensor-actuator network}
\newacronym{LLF-RC}{LLF-RC}{least-laxity-first with remaining constraints}
\newacronym{OA}{OA}{opportunistic aggregation}
\newacronym{RS}{RS}{repetitive scheduling}
\newacronym{HS}{HS}{hyperperiod scheduling}
\newacronym{TSCH}{TSCH}{time-slotted channel hopping}

\newacronym{EDF}{EDF}{earliest deadline first}
\newacronym{LLF}{LLF}{least-laxity-first}
\newacronym{CLLF}{CLLF}{conflict-aware least-laxity-first}
\newacronym{PC-LLF}{PC-LLF}{path-collision aware least-laxity-first}
\newacronym{RTQS}{RTQS}{real-time query scheduling}
\newacronym{DM}{DM}{deadline monotonic}
\newacronym{PDM}{PDM}{proportional deadline monotonic}

\newacronym{SPRF}{SPRF}{scheduling of periodic real-time flow}
\newacronym{TASA}{TASA}{traffic-aware scheduling algorithm}
\newacronym{RANDOM}{RANDOM}{random links scheduling algorithm}
\newacronym{MIS}{MIS}{maximum independent set}
\newacronym{MDCG}{MDCG}{multi-dimensional conflict graph}
\newacronym{MWIS}{MWIS}{maximum weighted independent set}
\newacronym{HSA}{HSA}{hop-wise scheduling algorithm}

\newacronym{SRCA}{SRCA}{slot reallocation for collision avoidance}
\newacronym{CBPF}{CBPF}{contention-based proportional fairness}
\newacronym{SmarTiSCH}{SmarTiSCH}{SmarTiSCH}

\newacronym{RSP}{RSP}{resource sharing path}
\newacronym{ASAP}{ASAP}{adaptation of slot-size and aggregation of packets}

\newacronym{SNR}{SNR}{signal-to-noise ratio}
\newacronym{SER}{SER}{symbol error rate}
\newacronym{PRR}{PRR}{packet reception rate}

\newacronym{RM}{RM}{rate monotonic}
\newacronym{EDZL}{EDZL}{earliest deadline zero laxity}
\newacronym{EPD}{EPD}{earliest proportional deadline}

\newacronym{LOS}{LOS}{line-of-sight}
\newacronym{NLOS}{NLOS}{non-line-of-sight}

\newacronym{CDF}{CDF}{cumulative distribution function}

%% file: preamble/abstract.tex

\begin{abstract}

This paper investigates scheduling strategies for wireless sensor-actuator networks (WSANs) in Industry 4.0 scenarios. 
In particular, we address the problem of real-time scheduling for multi-rate control systems by proposing a novel framework. 
Our framework features four strategies that improve reliability, schedulability and execution time, and reduce communication and storage costs. 
\emph{Two-phase scheduling} is our first strategy, devised to improve communication reliability. 
Our second strategy is the \emph{least-laxity-first with remaining conflicts} (LLF-RC) scheduling algorithm, which has high schedulability and affordable execution time. LLF-RC also keeps the maximum queue length at a moderate level, making it suitable for storage-constrained devices. 
Our third and fourth strategies are \emph{opportunistic aggregation} and \emph{repetitive scheduling}. Opportunistic aggregation performs simple and effective packet aggregation, enhancing schedulability by up to 97\% and reducing execution time by up to 29\%, in our simulation. 
Repetitive scheduling has negligible execution time, and contributes to minimize communication and storage costs. It reduces the maximum execution time by 92\% and the maximum communication and storage cost by 99\%, in our simulation.
We compare our proposed framework against existing approaches, and evaluate the advantages of our strategies in realistic scenarios.

\end{abstract}

%% file: preamble/keywords.tex


%% file: sections/introduction.tex

\section{Introduction} \label{section:introduction}

Industry 4.0 has become a major research focus, promising to revolutionize the industrial landscape by enhancing efficiency, productivity, adaptability, resilience, and overall process performance. A key aspect of this transformation is the shift from wired communication technologies, such as Fieldbus and real-time Ethernet \cite{thomesse2005:fieldbus-technology-industrial-automation, moyne2007:emergence-industrial-control-networks-manufacturing-control-diagnostics-safety-data}, to wireless alternatives, including \gls{TSCH} \cite{tsch12}, WirelessHART \cite{HART07}, and ISA100.11a \cite{ISA08}. Wireless connectivity overcomes the technical and economic constraints of wired systems, enabling rapid facility reconfiguration, reduced deployment and maintenance costs, and support for mobility-driven applications such as mobile robots \cite{rodriguez2021:5g-swarm-production-advanced-industrial-manufacturing-concepts-enabled-wireless-automation}, thereby paving the way for autonomous, smarter factories.

While wireless technologies offer substantial advantages, they also introduce challenges. Unlike the physical isolation of wired networks, wireless systems share the propagation medium, leading to potential interference and reliability concerns. Wireless solutions also tend to have lower throughput and higher latency due to retransmissions, power limitation, and bandwidth constraints. 

\begin{figure}
	\centering
	\begin{tikzpicture}[auto, thick, scale=0.8, every node/.style={scale=0.8}]
	  \node[cloud, fill=gray!20, cloud puffs=17, cloud puff arc= 110,
	        minimum width=7cm, minimum height=2.9cm, aspect=1] at (0,0) {};
	
	  \foreach \place/\x in {{(-2.5,0.3)/1}, 
	    {(-0.75,-0.7)/4}, {(1,0.5)/5}, {(0.25,0.7)/6},{(2.5,0.4)/9}}
	  \node[sensor] (s\x) at \place {};

	  \foreach \place/\y in {{(-1.2,-0.1)/1}, {(1.8,0.75)/2}, {(0.4,-0.7)/3},{(2.5,-0.4)/4}}
	  \node[actuator] (a\y) at \place {};

	  \foreach \place/\z in {{(0.75,-0.3)/7}, {(1.5,0)/8}, {(-1.75,-0.55)/2}, {(-1.2,0.55)/3}}
	  \node[relay] (r\z) at \place {};
	  
	  \node[controller] (c1) at (-0.05,0) {};

	  \path[thin, dashed] (s1) edge (r2);
	  \path[thin, dashed] (s1) edge (r3);
	  \path[thin, dashed] (r2) edge (r3);
	  \path[thin, dashed] (r3) edge (s6);
	  \path[thin, dashed] (r2) edge (s4);	
	    
	  \path[thin, dashed] (c1) edge (a1);
	  \path[thin, dashed] (c1) edge (s5);
	  \path[thin, dashed] (c1) edge (a3);

	  \path[thin, dashed] (c1) edge (s6);
	  \path[thin, dashed] (c1) edge (r7);
	  \path[thin, dashed] (r7) edge (r8);
	  \path[thin, dashed] (r8) edge (s9);
	  \path[thin, dashed] (r8) edge (a4);
	  \path[thin, dashed] (r8) edge (a2);

  	  \node[sensor] at (4.5,1) {};	 
  	  \node[actuator] at (4.5,0.5) {};	
  	  \node[controller] at (4.5,0) {};
  	  \node[relay] at (4.5,-0.5) {};
  	  \node at (4.5,-1) {--};

  	  \node[anchor=west] at (5,1) {\footnotesize Sensor};
  	  \node[anchor=west] at (5,0.5) {\footnotesize Actuator};
  	  \node[anchor=west] at (5,0) {\footnotesize Controller};
  	  \node[anchor=west] at (5,-0.5) {\footnotesize Relay};
  	  \node[anchor=west] at (5,-1) {\footnotesize Wireless link};
			  	
	\end{tikzpicture}
	\caption{WSAN consisting of one controller and multiple sensors, actuators, and relays.} 
	\label{fig:wsan}
	\vspace{-0.5cm}
\end{figure}
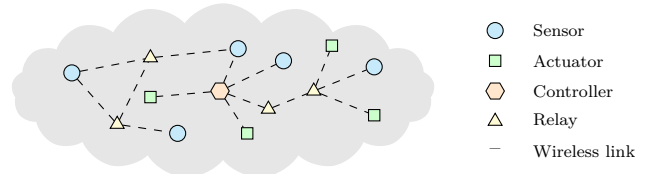






As technology advances, industrial processes become increasingly demanding and sophisticated. The most challenging scenarios will require complex automation relying on numerous sensors and actuators communicating wirelessly, as illustrated in \cref{fig:wsan}. Such a collection of interconnected devices, known as \gls{WSAN}, forms the communication backbone for industrial processes.

A cornerstone of industrial automation is real-time periodic control, implemented through control loops that continuously sense, compute, and actuate at fixed intervals, as shown in Fig. \ref{fig:block-diagram}. Particularly, the network controller computes and broadcasts the data flow schedule to all motes (sensors, relays, and actuators) via gateways. Control packets (controller $ \rightarrow $ actuators) and sensor packets (sensors $ \rightarrow $ controller) are then transmitted according to this schedule, often relayed through relays. 

The heterogeneity of data rates produced by control loops in \glspl{WSAN} significantly increases the risk of packet collisions and complicate the wireless data flow scheduling, ultimately impacting communication performance. Therefore, ensuring reliable, real-time communication is critical. At the same time, there is a pressing need for scheduling algorithms that are not only efficient but also lightweight, minimizing both storage requirements and communication overhead to maintain the scalability of the network. Storage is particularly vital due to the memory constraints of sensors and actuators, while communication costs impact time and frequency resources. These challenges highlight the need for lightweight and efficient scheduling strategies.

Based on the above motivation, it is essential to design a \gls{WSAN} scheduling framework that ensures high schedulability, high reliability, low execution time, low storage constraints, and low communication overhead. To the best of our knowledge, the scheduling of \glspl{WSAN} for multi-rate periodic control systems remains an open problem. Motivated by this, the present work introduces a framework with several novel strategies to meet these objectives.
Our contributions are summarized below:

\begin{figure*}[!t] 
   \centering
   \includegraphics[width=14cm]{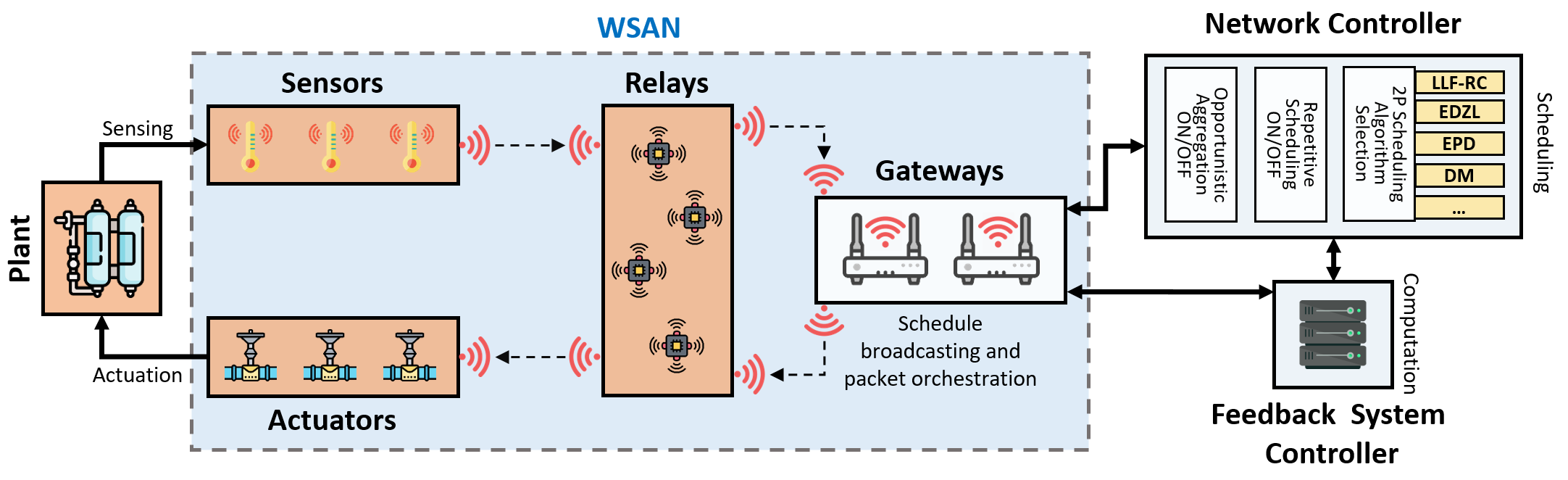}
	\vspace{0.25cm}
	\caption{Block diagram of WSAN and control loop. 
	\emph{Control loops form the core of periodic real-time systems, continuously executing sensing, computation, and actuation cycles. Industrial processes often involve multiple loops with heterogeneous rates, where feedback controllers manage actuators. A network controller schedules the wireless network according to each loop's timing requirements, which is the focus of this paper. We propose novel scheduling algorithms LLF-RC, opportunistic aggregation and repetitive scheduling, which work in synergy and significantly improve network schedulability, algorithm execution time, and memory and communication overhead.
	}
	}
	\label{fig:block-diagram}
	\vspace{-0.5cm}
\end{figure*}

\begin{itemize} [leftmargin = 0.3cm] 

	\item We propose a centralized \gls{TDMA}-based \gls{WSAN} scheduling framework, drawing inspiration from multiprocessor scheduling principles. It adapts these principles to address the unique challenges of \glspl{WSAN}. Our framework includes \emph{two-phase (2P) scheduling}, which is a novel idea that improves reliability compared to traditional scheduling, referred to as \emph{one-phase (1P) scheduling} in this paper. We demonstrate analytically, and through simulations that \emph{2P scheduling} has higher reliability, albeit with an increased delay compared to \emph{1P scheduling}. 
	
	\item We propose the \gls{LLF-RC} algorithm, which improves the scheduling performance by accounting for  laxity to  scheduling deadline and link bottleneck. Through extensive simulations, we demonstrate that \gls{LLF-RC} has higher schedulability and lower execution time compared to other scheduling algorithms.
	
	\item We propose \gls{OA}, a mechanism that interfaces seamlessly with any scheduling algorithm, and significantly increases the schedulability of WSANs.
	
	\item We propose \gls{RS}, a mechanism that can be integrated with any scheduling algorithm on condition that harmonic periods are employed. It is very effective in reducing communication and storage costs, with minimal penalty to schedulability. These characteristics make \gls{RS} ideal for resource-constrained wireless nodes.
	
	\item We adapt existing scheduling algorithms to work within the framework we propose in order to perform a fair comparison.
	
\end{itemize}

This paper is organized as follows: \cref{sec:related-work} reviews related work on wireless scheduling algorithms. \cref{sec:system-model} details the system model, including topology, communication, and routing. \cref{sec:proposed-scheduling-framework} introduces our \emph{2P scheduling} framework, along with the LLF-RC algorithm, \gls{OA}, and \gls{RS}. \cref{sec:extended-framework} adapts various multiprocessor scheduling algorithms to our framework. \cref{sec:simulations} discusses simulation results, \cref{sec:extension-large-net} discusses applicability to large-scale real-world \glspl{WSAN} and \cref{sec:conclusion} presents our conclusions.

\emph{Notation:} In this paper, $ \mathbb{Z}^+ $ denotes the set of positive integers. The functions $ \mathrm{max} \left\lbrace \mathcal{S} \right\rbrace $ and $ \mathrm{min} \left\lbrace \mathcal{S} \right\rbrace $ return the maximum and minimum value of set $ \mathcal{S} $, respectively. The product operator is denoted by $ \Pi \left( \cdot \right) $ and the logarithm in base $ 10 $ of a real number $ x > 0 $ is represented by $ \mathrm{log}_{10} (x) $.

%% file: sections/related-work.tex

\section{Related Work} \label{sec:related-work}

\begin{table*}[!t]
	\fontsize{6}{6}\selectfont
	\setlength\tabcolsep{2.9pt}
	\centering
	\caption{Categorization of relevant related work}
	\label{table_related_literature}
	\begin{tabular}{c c c c c c c c c c c c c c c c c}
		\toprule
		\multirow{3}{*}{\centering{\makecell{Approach}}} &
		\multirow{3}{*}{\centering{\makecell{Type}}} &
		\multirow{3}{*}{\centering{\makecell{Sch.}}} &
		\multirow{3}{*}{\centering{\makecell{Use \\ case}}} & 
		\multirow{3}{*}{\centering{\makecell{Algorithm}}} &
		\multirow{3}{*}{\centering{\makecell{Optional \\ features}}} & 
		\multirow{3}{*}{\centering{\makecell{Spatial \\ reuse}}} & 
		\multirow{3}{*}{\centering{\makecell{Number \\ of flows}}} & 
		\multicolumn{2}{c}{Traffic} & 
		\multicolumn{2}{c}{Deadline} & 
		\multicolumn{5}{c}{Performance} \\ 
		\cmidrule(lr){9-10}
		\cmidrule(lr){11-12}
		\cmidrule(lr){13-17}
		& & & & & & & &
		Periodic & 
		Aperiodic & 
		\makecell{Hard} & 
		\makecell{Soft} & 
		\makecell{Throughput} & 
		\makecell{Reliability} & 
		\makecell{Schedulability} & 
		\makecell{Storage \\ cost} & 
		\makecell{Execution \\ time} \\ 
		\midrule
		\midrule
		
		{\centering \cite{caccamo2002:an-implicit-prioritized-access-protocol-wireless-sensor-networks}} & 
		{\centering C} &
		{\centering 1P} &
		{\centering Generic} &
		{\centering EDF} &
		{\centering N/A} &
		{\centering Yes} &
		{\centering Single} & 
		{\centering Yes} &
		{\centering Yes} & 
		{\centering Yes} & 
		{\centering Yes} & 
		{\centering High} & 
		{\centering Medium} & 
		{\centering Medium} & 
		{\centering Medium} & 
		{\centering Low} \\

		{\centering \cite{chipara2007:real-time-query-scheduling-wireless-sensor-networks}} & 
		{\centering D} &
		{\centering 1P} &
		{\centering Generic} &
		{\centering RTQS \& DCQS} &
		{\centering PQS, NQS, SSQ} &
		{\centering Yes} &
		{\centering Single} & 
		{\centering Yes} &
		{\centering No} & 
		{\centering Yes} & 
		{\centering {No}} & 
		{\centering High} & 
		{\centering High} & 
		{\centering Medium} & 
		{\centering High} & 
		{\centering Medium} \\

		{\centering \cite{saifullah2010:real-time-scheduling-wirelesshart-networks}} & 
		{\centering C} &
		{\centering 1P} &
		{\centering Industry} &
		{\centering CLLF} &
		{\centering N/A} &
		{\centering No} &
		{\centering Multiple} & 
		{\centering Yes} &
		{\centering No} & 
		{\centering Yes} & 
		{\centering No} & 
		{\centering Medium} & 
		{\centering High} & 
		{\centering Very high} & 
		{\centering High} & 
		{\centering Medium} \\

		{\centering \cite{darbandi2019:path-collision-aware-real-time-link-scheduling-tsch-wireless-networks}} & 
		{\centering C} &
		{\centering 1P } &
		{\centering Industry} &
		{\centering PC-LLF} &
		{\centering N/A} &
		{\centering Yes} &
		{\centering Multiple} & 
		{\centering Yes} &
		{\centering No} & 
		{\centering Yes} & 
		{\centering No} & 
		{\centering Low} & 
		{\centering High} & 
		{\centering High} & 
		{\centering High} & 
		{\centering Very high} \\

		{\centering \cite{palattella2012:traffic-aware-scheduling-algorithm-reliable-low-power-multihop-ieee-802-15-4e-networks}} & 
		{\centering C} &
		{\centering 1P} &
		{\centering Industry} &
		{\centering Matching and coloring} &
		{\centering N/A} &
		{\centering Yes} &
		{\centering Multiple} & 
		{\centering Yes} &
		{\centering No} & 
		{\centering Yes} & 
		{\centering No} & 
		{\centering Low} & 
		{\centering High} & 
		{\centering High} & 
		{\centering Medium} & 
		{\centering High} \\

		{\centering \cite{shi2019:transmission-scheduling-periodic-real-time-traffic-ieee-802-15-4e-tsch-based-industrial-mesh-networks}} & 
		{\centering C} &
		{\centering 1P} &
		{\centering Industry} &
		{\centering SPRF} &
		{\centering N/A} &
		{\centering Yes} &
		{\centering Multiple} & 
		{\centering Yes} &
		{\centering No} & 
		{\centering Yes} & 
		{\centering No} & 
		{\centering Low} & 
		{\centering High} & 
		{\centering High} & 
		{\centering High} & 
		{\centering High} \\

		{\centering \cite{chen2019:joint-scheduling-channel-allocation-end-to-end-delay-minimization-industrial-wirelesshart-networks}} & 
		{\centering C} &
		{\centering 1P} &
		{\centering Industry} &
		{\centering HSA} &
		{\centering N/A} &
		{\centering No} &
		{\centering Multiple} & 
		{\centering Yes} &
		{\centering No} & 
		{\centering Yes} & 
		{\centering No} & 
		{\centering Low} & 
		{\centering High} & 
		{\centering High} & 
		{\centering High} & 
		{\centering Very high} \\

		{\centering \cite{park2021:slot-reallocation-rejection-collision-avoidance-autonomous-tsch-networks}} & 
		{\centering A} &
		{\centering 1P} &
		{\centering Generic} &
		{\centering SRCA} &
		{\centering N/A} &
		{\centering Yes} &
		{\centering Single} & 
		{\centering Yes} &
		{\centering No} & 
		{\centering No} & 
		{\centering Yes} & 
		{\centering Low} & 
		{\centering High} & 
		{\centering Medium} & 
		{\centering Medium} & 
		{\centering Low} \\

		{\centering \cite{bommisetty2023:contention-based-proportional-fairness-transmission-scheme-time-slotted-channel-hopping-networks}} & 
		{\centering C} &
		{\centering 1P} &
		{\centering Industry} &
		{\centering CBPF} &
		{\centering N/A} &
		{\centering Yes} &
		{\centering Single} & 
		{\centering Yes} &
		{\centering No} & 
		{\centering No} & 
		{\centering Yes} & 
		{\centering Very high} & 
		{\centering High} & 
		{\centering Very high} & 
		{\centering Very high} & 
		{\centering Very high} \\

		{\centering \cite{yu2022:smartisch-interference-aware-engine-ieee-802-15-4e-based-networks}} & 
		{\centering A} &
		{\centering 1P} &
		{\centering Industry} &
		{\centering SmarTiSCH} &
		{\centering N/A} &
		{\centering Yes} &
		{\centering Single} & 
		{\centering Yes} &
		{\centering No} & 
		{\centering No} & 
		{\centering Yes} & 
		{\centering Low} & 
		{\centering High} & 
		{\centering High} & 
		{\centering High} & 
		{\centering Very Low} \\

		{\centering \cite{moon2022}} & 
		{\centering C} &
		{\centering 1P} &
		{\centering Industry} &
		{\centering AdaptiveHART} &
		{\centering N/A} &
		{\centering No} &
		{\centering Few} & 
		{\centering Yes} &
		{\centering No} & 
		{\centering Yes} & 
		{\centering No} & 
		{\centering High} & 
		{\centering High} & 
		{\centering High} & 
		{\centering Medium} & 
		{\centering Medium} \\

		{\centering \cite{seo2023}} & 
		{\centering C} &
		{\centering 1P} &
		{\centering Industry} &
		{\centering RSP} &
		{\centering N/A} &
		{\centering Yes} &
		{\centering Many} & 
		{\centering Yes} &
		{\centering No} & 
		{\centering Yes} & 
		{\centering No} & 
		{\centering High} & 
		{\centering High} & 
		{\centering Very high} & 
		{\centering High} & 
		{\centering High} \\

		{\centering \cite{kim2024}} & 
		{\centering C} &
		{\centering 1P} &
		{\centering Industry} &
		{\centering ASAP} &
		{\centering N/A} &
		{\centering No} &
		{\centering Many} & 
		{\centering Yes} &
		{\centering No} & 
		{\centering Yes} & 
		{\centering No} & 
		{\centering High} & 
		{\centering High} & 
		{\centering Medium} & 
		{\centering High} & 
		{\centering Medium} \\

		\midrule
		
		{\centering Proposed} & 
		{\centering C} &
		{\centering 2P} &
		{\centering Industry} &
		{\centering LLF-RC} &
		{\centering -} &
		{\centering No} &
		{\centering Multiple} & 
		{\centering Yes} &
		{\centering No} & 
		{\centering Yes} & 
		{\centering No} & 
		{\centering Low} & 
		{\centering High} & 
		{\centering Medium} & 
		{\centering High} & 
		{\centering Medium} \\
		
		{\centering Proposed} & 
		{\centering C} &
		{\centering 2P} &
		{\centering Industry} &
		{\centering LLF-RC} &
		{\centering {OA}} &
		{\centering No} &
		{\centering Multiple} & 
		{\centering Yes} &
		{\centering No} & 
		{\centering Yes} & 
		{\centering No} & 
		{\centering High} & 
		{\centering High} & 
		{\centering High} & 
		{\centering High} & 
		{\centering Medium} \\
		
		{\centering Proposed} & 
		{\centering C} &
		{\centering 2P} &
		{\centering Industry} &
		{\centering LLF-RC} &
		{\centering RS} &
		{\centering No} &
		{\centering Multiple} & 
		{\centering Yes} &
		{\centering No} & 
		{\centering Yes} & 
		{\centering No} & 
		{\centering Low} & 
		{\centering High} & 
		{\centering Medium} & 
		{\centering Low} & 
		{\centering Low} \\
		
		{\centering Proposed} & 
		{\centering C} &
		{\centering 2P} &
		{\centering Industry} &
		{\centering LLF-RC} &
		{\centering OA+RS} &
		{\centering No} &
		{\centering Multiple} & 
		{\centering Yes} &
		{\centering No} & 
		{\centering Yes} & 
		{\centering No} & 
		{\centering High} & 
		{\centering High} & 
		{\centering High} & 
		{\centering Low} & 
		{\centering Low} \\

		\bottomrule
		\multicolumn{5}{l}{\tiny C: Centralized ~~~~ D: Distributed ~~~~ A: Autonomous}

	\end{tabular}
	\label{tab:related-literature}
	\vspace{1mm}\\
\end{table*}

The most prominent \gls{WSAN} protocols,  including \gls{TSCH} \cite{tsch12}, WirelessHART \cite{HART07}, and ISA100.11a \cite{ISA08}, build on the IEEE 802.15.4 standard \cite{DOT15_4}, which uses \gls{TDMA} across multiple channels. While similar at the physical layer, they differ in channel reuse policies. Specifically, \gls{TSCH} allows it, whereas WirelessHART and ISA100.11a do not. All three protocols allow custom scheduling designs, which have led to a broad literature on scheduling algorithms.

Scheduling approaches can be categorized as centralized, distributed, or autonomous. Centralized algorithms (e.g., \cite{saifullah2010:real-time-scheduling-wirelesshart-networks, palattella2012:traffic-aware-scheduling-algorithm-reliable-low-power-multihop-ieee-802-15-4e-networks, jin2016:centralized-scheduling-algorithm-ieee-802-15-4e-tsch-industrial-low-power-wireless-networks}) rely on a controller with global knowledge. Distributed solutions (e.g., \cite{tinka2011:decentralized-scheduling-algorithm-time-synchronized-channel-hopping, accettura2013:decentralized-traffic-aware-scheduling-multi-hop-low-power-lossy-networks-internet-things, palattella2016:on-the-fly-bandwidth-reservation-6tisch-wireless-industrial-networks, gomes2017:mabo-tsch-multihop-blacklist-based-optimized-time-synchronized-channel-hopping, hamza2019:enhanced-minimal-scheduling-function-ieee-802-15-4e-tsch-networks, domingo2016:distributed-pid-based-scheduling-6tisch-networks}) requires individual nodes to compute partial schedules by exchanging information with neighbors. Autonomous methods (e.g., \cite{duquennoy2015:orchestra-robust-mesh-networks-autonomously-scheduled-tsch, kim2019:alice-autonomous-link-based-cell-scheduling-tsch}) let nodes derive schedules locally, often based on routing information. Centralized algorithms excel at optimizing periodic traffic, while distributed and autonomous algorithms are more resilient to changes in network topology and traffic patterns. In the following, we revisit the most relevant works.

Early works in networked control, such as \cite{caccamo2002:an-implicit-prioritized-access-protocol-wireless-sensor-networks}, addressed deadline-aware scheduling in cell-based networks, combining \gls{EDF} scheduling within cells and priority-based inter-cell coordination. Unused slots from hard-deadline flows were opportunistically reassigned to aperiodic traffic. Similarly, \gls{RTQS} \cite{chipara2007:real-time-query-scheduling-wireless-sensor-networks} tackled scheduling over arbitrary interference graphs, introducing three variants, namely, preemptive (PQS), non-preemptive (NQS), and slack stealing query scheduling (SSQ), that balanced throughput and deadline compliance.

Closer to our focus, a large body of work has investigated scheduling multiple real-time flows, where each flow carries deadline constraints \cite{saifullah2010:real-time-scheduling-wirelesshart-networks, saifullah2011:end-to-end-delay-analysis-fixed-priority-scheduling-wirelesshart-networks, saifullah2011:priority-assignment-real-time-flows-wirelesshart-networks, shi2019:transmission-scheduling-periodic-real-time-traffic-ieee-802-15-4e-tsch-based-industrial-mesh-networks, darbandi2019:path-collision-aware-real-time-link-scheduling-tsch-wireless-networks, chen2019:joint-scheduling-channel-allocation-end-to-end-delay-minimization-industrial-wirelesshart-networks, park2021:slot-reallocation-rejection-collision-avoidance-autonomous-tsch-networks, bommisetty2023:contention-based-proportional-fairness-transmission-scheme-time-slotted-channel-hopping-networks, yu2022:smartisch-interference-aware-engine-ieee-802-15-4e-based-networks,  kim2019:alice-autonomous-link-based-cell-scheduling-tsch, moon2022, seo2023, kim2024}.

Specifically, \cite{saifullah2010:real-time-scheduling-wirelesshart-networks} mapped control loops to flows, proving WirelessHART scheduling is NP-hard and proposing an optimal branch-and-bound method. However, its computational complexity limited scalability. The lower-complexity \gls{CLLF} algorithm was subsequently introduced, offering a better performance-complexity tradeoff than \cite{saifullah2010:real-time-scheduling-wirelesshart-networks}, and significantly outperforming \gls{DM} and \gls{PDM} algorithms. Later, \cite{darbandi2019:path-collision-aware-real-time-link-scheduling-tsch-wireless-networks} proposed \gls{PC-LLF}, prioritizing path-level conflicts over next-link conflicts, achieving higher schedulability than \gls{CLLF}. Besides, \cite{saifullah2011:end-to-end-delay-analysis-fixed-priority-scheduling-wirelesshart-networks} analyzed end-to-end delays for fixed-priority \gls{DM} and \gls{PDM}, while \cite{saifullah2011:priority-assignment-real-time-flows-wirelesshart-networks} presented worst-case delay analysis and a heuristic to minimize scheduling overhead.

\Gls{TASA} \cite{palattella2012:traffic-aware-scheduling-algorithm-reliable-low-power-multihop-ieee-802-15-4e-networks} minimized slots through matching and coloring, inspiring \gls{SPRF} \cite{shi2019:transmission-scheduling-periodic-real-time-traffic-ieee-802-15-4e-tsch-based-industrial-mesh-networks}, which prioritized periodic flows but incurred high computation overhead due to repeated matchings.

Other works tackled flow delay minimization. For instance, \gls{HSA} \cite{chen2019:joint-scheduling-channel-allocation-end-to-end-delay-minimization-industrial-wirelesshart-networks} modeled scheduling as a multi-dimensional conflict graph problem, solved via approximations but at growing computational cost. In contrast, \gls{SRCA} algorithm \cite{park2021:slot-reallocation-rejection-collision-avoidance-autonomous-tsch-networks} autonomously updated schedules at each node using parent-child relations. By preemptively changing a child node's slot after a successful transmission, \gls{SRCA} avoided future collisions, boosting reliability.

Efforts like \gls{CBPF} \cite{bommisetty2023:contention-based-proportional-fairness-transmission-scheme-time-slotted-channel-hopping-networks} applied convex optimization to reduce collisions, while SmarTiSCH \cite{yu2022:smartisch-interference-aware-engine-ieee-802-15-4e-based-networks} enabled passive interference detection without additional control overhead. 
Further enhancing adaptability, AdaptiveHART \cite{moon2022} dynamically adjusted scheduling and transmission priorities in response to changing traffic conditions. Unlike traditional static WirelessHART, it optimized time slot allocation, reducing latency and packet loss while improving resource utilization. Similarly, \gls{RSP} \cite{seo2023} enhanced efficiency by allowing flows to share paths without interference, minimizing delays and contention. For finer-grained adaptation, \gls{ASAP} \cite{kim2024} dynamically adjusted slot sizes and aggregateed smaller packets when feasible. By prioritizing critical data and reducing overhead, \gls{ASAP} allowed balancing efficiency and reliability, making it particularly effective in highly dynamic environments. Together, these algorithms, \gls{CBPF}, SmarTiSCH, AdaptiveHART, \gls{RSP}, and \gls{ASAP}, demonstrated a progression from collision mitigation to adaptive scheduling, addressing key  challenges in real-time industrial wireless networks.

While settings like vehicular networks and smart grids share some similarities with industrial \glspl{WSAN}, they often differ in complexity and specific requirements. These environments may lack actuators, omit multi-rate control loops, or impose less stringent demands. Nevertheless, we revisit relevant works in these domains to provide broader context.

In vehicular safety networks, vehicles broadcast messages to share location and maneuver data. Several studies have explored such scenarios \cite{vehicular1, vehicular2, vehicular3, vehicular4}, where scheduling plays a crucial role in ensuring high reliability and schedulability. While received messages may influence actuator control within vehicles, resembling certain aspects of \gls{WSAN}, scheduling is typically centralized at base stations with substantially greater computational resources than industrial gateways. Moreover, safety-critical vehicular networks operate as single-hop systems to minimize latency, eliminating the need for onboard scheduling storage. Based on these facts, execution time and storage cost are less critical considerations compared to \glspl{WSAN}. Additionally, unlike \glspl{WSAN}, vehicle controllers and actuators usually connect via wired interfaces, reducing failure points and improving overall reliability.

Smart grids employ \glspl{WSAN} under stringent operational constraints, comparable to industrial settings. In these scenarios, \glspl{WSAN} are used to monitor and control critical grid infrastructure, enabling seamless communication between distributed nodes for fault detection, load balancing, and real-time control. As such, effective scheduling mechanisms are essential for ensuring reliability and performance \cite{grid1,grid2,grid3}. However, due to the critical nature of smart grids, wireless technology is typically limited to small subnetworks, while the most critical functions are still predominantly managed through wired connections.

While inspired by industrial applications, our proposed framework addresses challenges not tackled by prior works, including those in vehicular and smart grid contexts.
For clarity, \cref{tab:related-literature} summarizes the discussed studies, highlighting their advantages and disadvantages. Other relevant works worth mentioning include \cite{lee2021:link-based-autonomous-cell-scheduling-ieee-802.15.4e-tsch-improved-traffic-throughput, sinha2021:deadline-aware-scheduling-maximizing-information-freshness-industrial-cyber-physical-system, kim2024:slot-size-adaptation-utility-based-packet-aggregation-ieee-802-15-4e-time-slotted-communication-networks, ergen2010:tdma-scheduling-algorithms-wireless-sensor-networks, yuan2012:tree-based-multi-channel-convergecast-wireless-sensor-networks}.

%% file: sections/system-model.tex

\section{System Model} \label{sec:system-model}

This section introduces the link quality, network topology, communication, and routing models for the considered \gls{WSAN}.

\subsection{Link Quality Model} \label{sec:link-quality-model}

Simulating link quality requires the use a realistic radio propagation model. Therefore, we employ the log-normal path-loss model due to its accuracy to depict large-scale fading, which is typical in industrial environments \cite{rappaport1989:uhf-fading-factories, zuniga2004:analyzing-transitional-region-low-power-wireless-links, tanghe2008:industrial-indoor-channel-large-scale-temporal-fading-900-2400-5200-mhz}. The model, expressed in decibels (dB), is defined as 
\begin{equation}
	\mathrm{PL}(d) = \mathrm{PL}(d_0) + 10 \eta \cdot \log_{10} \left( \frac{d}{d_0} \right) + X_{\sigma},
\end{equation} 
where $ \mathrm{PL}(d) $ is the path-loss of the signal strength at distance $ d $, $ \mathrm{PL}(d_0) $ is the path-loss at reference distance $ d_0 $, $\eta$ is the path-loss exponent, and $ X_{\sigma} $ is a zero-mean Gaussian random variable, which represents shadowing, and has standard deviation $ \sigma $ \cite{rappaport2001:wireless-communications-principles-practice}. The received power $ P(d) $ at distance $ d $ is given by
\begin{equation}
	P(d) = P_\mathrm{tx} - \mathrm{PL}(d),
\end{equation}
where $ P_\mathrm{tx} $ is the transmit power. Additionally, the \gls{SNR} is defined as
\begin{equation}
	\gamma = P(d) - P_\mathrm{n},
\end{equation}
where $ P_\mathrm{n} $ is the noise floor. Furthermore, the \gls{SER} is expressed as
\begin{equation}
	\mathrm{SER} = \frac{1}{2} \mathrm{erfc} \bigg( \frac{\beta_1  (\gamma - \beta_2)}{\sqrt{2}} \bigg),
\end{equation}
based on the empirically determined TOSSIM model \cite{levis2003:tossim-accurate-scalable-simulation-entire-tiny-os-applications}, which is employed to characterize the \gls{PRR}, given by
\begin{equation} \label{eqn:prr}
	\mathrm{PRR} = \left( 1 - \mathrm{SER} \right)^{2 L },
\end{equation}
where $ L $ is the packet length measured in bytes.

\subsection{Network Topology Model}

We model the topology of the wireless network as an undirected simple graph $ \mathcal{G = (V, E)} $, meaning the graph has no loops nor parallel edges. Set $ \mathcal{V} $ indexes the nodes while set $ \mathcal{E} $ indexes the undirected edges, representing the links between every pair of nodes. The nodes in the graph can represent gateways or motes. A mote can be a sensor, an actuator, or a pure relay, with the ability to relay packets. Thus, the set of nodes $ \mathcal{V} $ is defined as $ \mathcal{V} = \mathcal{N} \cup \mathcal{M} $, where $ \mathcal{N} $ is the set of motes, and $ \mathcal{M} $ is the set of gateways. Specifically, every edge $ e \in \mathcal{E} $ has a link quality given by the \gls{PRR}, defined in (\ref{eqn:prr}) and denoted by $ \mathrm{PRR}(e) \in \left[ 0, 1 \right] $.

\subsection{Communication Model}

In the following, we define relevant aspects of the communication model.

{\textbf{Multiplexing and duplexing.}} We represent the communication model as a discrete-time system as we rely on TDMA. We assume that a unit time is equal to the duration of a slot, and that transmitting a packet via a link requires one slot. Also, each gateway or mote is equipped with a half-duplex radio that cannot transmit and receive simultaneously.

{\textbf{Link.}} We assume that gateways are connected with perfect link quality and zero delay. This is a realistic assumption, since gateways in industrial settings are access points connected via fast and redundant wired links. However, the link between a gateway and a mote, or between two motes, is assumed to be imperfect because it is wireless. Specifically, the link $ e $ between two wireless nodes is viable only if its quality is greater than a threshold $ \Gamma_\mathrm{th} $, i.e. $ \mathrm{PRR}(e) \geq \Gamma_\mathrm{th} $. In addition, spatial reuse of channels is disabled as in WirelessHART, facilitating interference mitigation by preventing concurrent transmissions over the same channels.

{\textbf{Flow.}} We consider a system with multiple control loops. The controllers governing the control loops are located at the gateways, and the communication of a loop is called a flow. Each flow is activated periodically, and has a sensor and an actuator. In particular, the sensor samples, and transmits at times $ k \cdot p_f $, where $ k = \left\lbrace 0, 1, 2, ... \right\rbrace $, and $ p_f $ is the period of flow $ f $. A flow can have multiple routing paths, each starting at a sensor and ending at an actuator, going through a controller as shown in \cref{fig:routing}. Specifically, $ \pi_i^\mathrm{sc} $, $ i = \left\lbrace 0, \dots, n-1 \right\rbrace $, represents the $ i $-th sensor-to-controller path (sc-path) whereas $ \pi_j^\mathrm{ca} $, $ j = \left\lbrace 0, \dots, m-1 \right\rbrace $, represents the $ j $-th controller-to-actuator path (ca-path). In addition, $ n $ and $ m $ denote the number of sa- and ca-paths, respectively.

{\textbf{Flow delay.}} We employ the constant network-induced delay model \cite{zhang2001:stability-networked-control-systems}. In this model, there is no need to differentiate the sensor-to-controller delay $ d^\mathrm{sc} $ and the controller-to-actuator delay $ d^\mathrm{ca} $, if the control laws are time-invariant, which is true for our system. Hence, we use the \emph{sum flow delay} $ d = d^\mathrm{sc}  + d^\mathrm{ca} $ for analyzing the communication schedulability. The flow delay needs to satisfy $ d \le d_\mathrm{max} $, where $d_\mathrm{max}$ is the maximum allowable delay, determined by stability analysis and feedback loop performance.

{\textbf{Flow reliability.}} To guarantee stability and satisfactory performance for a feedback loop, the flow reliability $ r^f $ of flow $ f $ needs to be greater than a lower bound $ r_\mathrm{min} $, i.e. $ r^f \ge r_\mathrm{min} $. For control applications, usually $ r_\mathrm{min} $ does not need to be as high as in monitoring applications \cite{zhang2001:stability-networked-control-systems, marco2010:trend-timely-reliable-energy-efficient-dynamic-wsn-protocol-control-applications}.

\subsection{Routing Model} \label{sec:routing-model}
We consider multi-path routing for flow forwarding. In particular, when the reliability requirement of a flow $ f $ cannot be satisfied by a single sensor-to-actuator path (sa-path), i.e., $ r^f \leq r_\mathrm{min} $, multi-path routing helps to address this limitation. A flow $f$ is split into two parts: one comprises the paths starting at the sensor and ending at a controller, and the other comprises the paths starting at a controller and ending at the actuator. Specifically, a flow $ f $ is divided into a sensor-to-controller subflow (sc-flow), denoted by $f^\mathrm{sc} $, and a controller-to-actuator subflow (ca-flow), denoted by $ f^\mathrm{ca} $. The paths for $ f^\mathrm{sc} $ are called sc-paths, while the paths for $ f^\mathrm{ca} $ are called ca-paths, as shown in \cref{fig:routing}. When the sensor data arrives at the controller through an sc-path, the control algorithm is ready to be executed. Once the control algorithm is executed, the output is ready to be sent to the actuator via a ca-path.
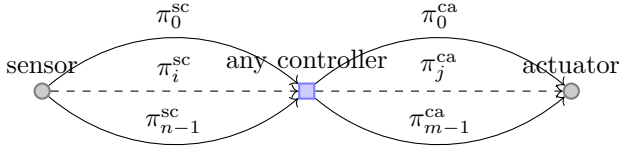
\begin{figure}[!t]
	\centering
	\scalebox{1} {
	\begin{tikzpicture}
		[GW/.style={rectangle,draw=blue!50,fill=blue!20,thick,
		inner sep=0pt,minimum size=2mm},
		node/.style={circle,draw=black!50,fill=black!20,thick,
		inner sep=0pt,minimum size=2mm}]
		\node at ( 0,0) [GW] (GW) {};
		\node at ( -3.5,0) [node] (src) {};
		\node at (3.5,0) [node] (dest) {};
		\draw [->,out=40,in=140] (src) to node [above] {$ \pi_0^\mathrm{sc} $} (GW);
		\draw [dashed,->,out=0,in=180] (src) to node [above] {$ \pi_{i}^\mathrm{sc} $} (GW);
		\draw [->,out=-40,in=-140] (src) to node [above] {$ \pi_{n-1}^\mathrm{sc} $} (GW);
		\draw [->,out=40,in=140] (GW) to node [above] {$ \pi_0^\mathrm{ca} $} (dest);
		\draw [dashed,->,out=0,in=180] (GW) to node [above] {$ \pi_j^\mathrm{ca} $} (dest);
		\draw [->,out=-40,in=-140] (GW) to node [above] {$ \pi_{m-1}^\mathrm{ca} $} (dest);
		\node [above] at (GW.north) {any controller};
		\node [above] at (src.north) {sensor};
		\node [above] at (dest.north) {actuator};
	\end{tikzpicture}
	}
\caption{A flow from sensor to actuator with $ n $ sc-paths $ \pi_0^\mathrm{sc}, \dots, \pi_{n-1}^\mathrm{sc} $ and $ m $ ca-paths $\pi_0^\mathrm{ca}, \dots, \pi_{m-1}^\mathrm{ca} $. The number of paths $ n $ and $ m $ are not necessarily equal.}
\label{fig:routing}
\vspace{-0.5cm}
\end{figure}

%% file: sections/proposed-framework.tex

\section{Proposed Scheduling Framework} \label{sec:proposed-scheduling-framework}
In this section, we first present our proposed framework. Then, we introduce \emph{2P scheduling} and \gls{LLF-RC} algorithm, devised to improve reliability and schedulability, respectively. In addition, we present \gls{OA} (Opportunistic Aggregation) and \gls{RS} (Repetitive Scheduling). Specifically, \gls{OA} improves schedulability, whereas \gls{RS} reduces execution time, and communication and storage costs. 
{In the end, we discuss the implications of applying the scheduling framework to industrial networks.}

\subsection{Proposed Scheduling Framework} 
Our proposed scheduling framework draws inspiration from the multiprocessor scheduling literature. Specifically, multiprocessor scheduling deals with the problem of assigning a set of tasks to a number of processors such that each task meets a given deadline. In particular, if there are $ C $ processors, the maximum number of tasks that can be simultaneously supported is $ C $. Inspired by this, our TDMA scheduling problem can be formulated as a multiprocessor scheduling problem, where tasks can be of different granularity.
This means that a task is mapped to a flow, an sc- or ca-path, or a link, while a processor is mapped to to an available frequency channel. Thus, given $ C $ channels, \gls{TDMA} scheduling allows at most $C$ transmissions in parallel if spatial reuse is disabled, as assumed in this work. Our \gls{TDMA} scheduling framework has one constraint that makes it different from multiprocessor scheduling: \emph{any two links scheduled in the same slot cannot be in primary conflict \cite{chlamtac1994:making-transmission-schedules-immune-topology-changes-multi-hop-packet-radio-networks}, meaning the links cannot share a common node}\footnote{This constraint is relaxed if \gls{OA} is employed, as explained in \cref{sec:opportunistic-aggregation}.}.

There are two types of algorithms in the multiprocessor scheduling literature: partitioned and global. The former assigns a task statically to a processor, while the latter allows tasks to migrate freely between processors \cite{davis2011:survey-hard-real-time-scheduling-multiprocessor-systems}. In our framework, we adopt global scheduling due to its better performance in deadline satisfaction. Making an analogy to multiprocessor scheduling, task migration is equivalent to scheduling two consecutive transmissions on two different channels, as a processor is mapped to a channel. Therefore, the cost of performing task migration is negligible, as a transceiver can receive a packet on one channel, and send the packet on another channel with almost no additional cost. \cref{alg:framework}, summarizes the proposed scheduling framework, which we explain in more detail in the following. 
\begin{algorithm}
	\KwIn{\\ $ C $: Number of available channels \\ $ B $: Maximum queue length} 
	\KwOut{\\ $ \mathcal{H}$: Scheduling table}
	Perform necessary schedulability test, including deadline check and utilization check\;
	\eIf {schedulability test is passed} {
	Set the hyperperiod length to the least common multiple of all flow periods, $ H = \lcm\{p_f: f \in \mathrm{all~flows} \}$\;
	\For {$ i = 0 $ \KwTo $ H - 1$} {
		Define the set $ \mathcal{J}_i $ with all released transmissions in slot $ i $, whose receiver queue length is lower than $ B $, and ordered by priority\;
		Perform a scheduling $ \mathcal{H}_i $ with at most $C$ non-conflicting transmissions from $ \mathcal{J}_i $, in descending order of priority\;
		\If{deadline is missed} {
			\Return $ \mathcal{H} = \emptyset $;		
		}
		\If{scheduling is complete} {
			\Return $ \mathcal{H} = \left\lbrace \mathcal{H}_j \right\rbrace_{j = 0}^{i} $;			
		}
	  }
	} {
		\Return $ \mathcal{H} = \emptyset $;		
	}
	\caption{Proposed scheduling framework}
	\label{alg:framework}
	\algorithmfootnote{The maximum queue length $ B $ is, in practice, related to the RAM size. Also, $ \mathrm{lcm} \left\lbrace \mathcal{S} \right\rbrace $ represents the least common multiple operation of a set $ \mathcal{S} $ whereas $\mathcal{H} = \emptyset $ indicates that a feasible scheduling was not found.}
\end{algorithm}

First, a schedulability test is performed (see \texttt{line 1}), which consists of a deadline check and a utilization check.

\textbf{Deadline check.} The deadline of a flow cannot be smaller than the minimum flow delay $ d^\mathrm{2P} $. 

\textbf{Utilization check.}  The utilization of a flow $ f $ is defined as $ u_f = \cfrac{\mathrm{hops}(f)}{p_f} $. Specifically, $\mathrm{hops}(f)$ is the number of transmissions needed for each activation of flow $ f $, and $ p_f $ is the period of flow $ f $. The total utilization of all flows should not be larger than the channel count $ C $ \cite{horn1974:some-simple-scheduling-algorithms}, i.e., $ \sum_f u_f \le C $, which is a necessary condition for the existence of a feasible scheduling. Note that we do not apply this check when \gls{OA} is employed.

If the schedulability test is passed (see \texttt{line 2}), scheduling is attempted (see \texttt{lines 4-6}) assuming that for one activation of a flow, the transmissions on the ca-flow are not released until all transmissions on the sc-flow are finished. If the scheduling is complete, then a feasible solution has been found (see \texttt{lines 9-10}). If the scheduling has not been achieved by the deadline, the scheduling is infeasible (see \texttt{lines 7-8}). Although it is theoretically possible for a control loop to have a deadline larger than the period, we impose that the deadline is not larger than the period, because of the following two reasons.

\begin{itemize} [leftmargin = 0.3cm]
	\item If a control loop is feasible for a deadline larger than the period, it is also feasible for a larger period (see Fig. 7 of \cite{zhang2001:stability-networked-control-systems} which shows the relation between stability region, period, and delay).
	
	\item A set of synchronized periodic tasks is schedulable if it is feasible for a hyperperiod, which is defined as the least common multiple of all task periods \cite{cucu2006:feasibility-intervals-fixed-priority-real-time-scheduling-uniform-multiprocessors}. Otherwise, even though the scheduling may still be periodic from some point in time, the point cannot be accurately determined \cite{cucu2007:feasibility-intervals-multiprocessor-fixed-priority-scheduling-arbitrary-deadline-periodic-systems}. 
\end{itemize}

\emph{\textsc{Remark:} Our proposed framework can be applied with any scheduling algorithm. The main difference is how released transmissions are prioritized (see \texttt{lines 5-6}). In particular, each scheduling algorithm handles released transmissions in a different way, resulting in different performances.}

\subsection{Proposed 2P Scheduling} \label{sec:proposed-2p-scheduling}

\cref{fig:1p-2p-scheduling} illustrates \emph{1P scheduling} and the proposed \emph{2P scheduling}. In particular, \emph{2P scheduling} executes the control algorithm only after the transmissions on all or some sc-paths have arrived to the controller by a given time. Then, the controller output is sent on all ca-paths, as shown in Fig. \ref{fig:1p-2p-scheduling-a}. On the contrary, \emph{1P scheduling} manages the communication of each sa-path independently, without synchronization of the controller, as shown  in Fig. \ref{fig:1p-2p-scheduling-b}. Specifically, each sa-path is composed of a sc-path and a ca-path, i.e., $\pi_i = \pi_i^\mathrm{sc}  \cup \pi_i^\mathrm{ca}$. The sa-path concept is not needed in \emph{2P scheduling}, and bears no meaning when the number of sc-paths and ca-paths are different. However, we keep the concept of sa-path in \emph{2P scheduling} only for consistency with \emph{1P scheduling}. In \emph{2P scheduling}, the gateway end-points of the two paths may be different, since we assume full connectivity between gateways. 
\begin{figure}[!h] 
    \centering
    \begin{subfigure}{\columnwidth}
    \includegraphics[width=\columnwidth]{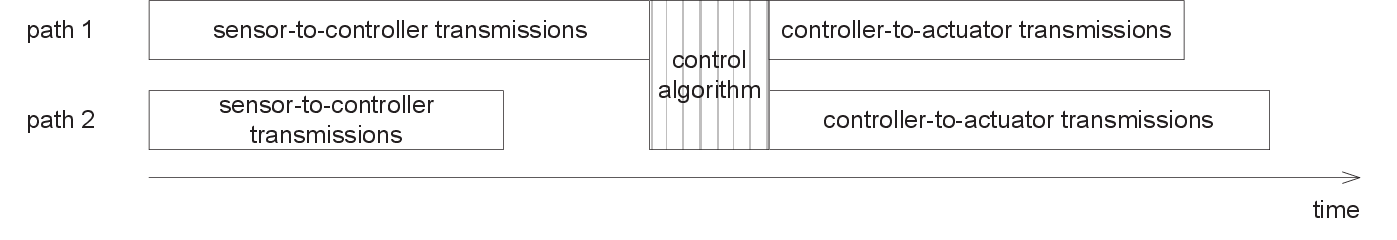}
    \caption{2P scheduling}
    \label{fig:1p-2p-scheduling-a}
    \vspace{0.3cm}
    \end{subfigure}
    \begin{subfigure}{\columnwidth}
    \includegraphics[width=\columnwidth]{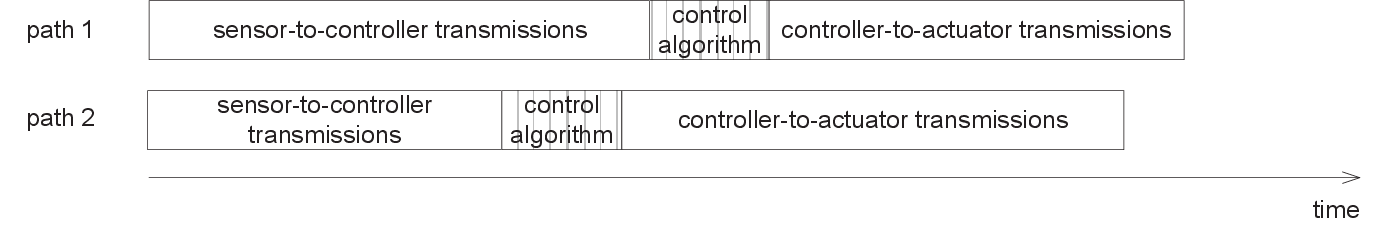}
    \caption{1P scheduling}
    \label{fig:1p-2p-scheduling-b}
    \end{subfigure}
    \caption{\emph{2P scheduling} and \emph{1P scheduling} are depicted for a WSAN with two gateways.}%
    \label{fig:1p-2p-scheduling}
\end{figure}

Assuming that a flow has $ n $ sa-paths, denoted by $ \pi_0, \pi_1, ..., \pi_{n-1} $, we will show that \emph{2P scheduling} provides higher communication reliability than \emph{1P scheduling}.


\begin{thm} \label{thm:reliability} Given a flow with $ n \geq 1 $ sa-paths, the end-to-end communication reliability of 2P scheduling is greater than or equal to that of 1P scheduling, i.e., $r^\mathrm{2P} \ge r^\mathrm{1P}$.
\end{thm}

\begin{proof}
The reader is referred to Appendix \ref{app:thm-1}.
\end{proof}

\begin{figure}[!h]
\centering
\scalebox{1} {
\begin{tikzpicture}
	[GW/.style={rectangle,draw=blue!50,fill=blue!20,thick,
	inner sep=0pt,minimum size=2mm},
	node/.style={circle,draw=black!50,fill=black!20,thick,
	inner sep=0pt,minimum size=2mm}]
	\node at ( 0,0) [GW] (GW) {};
	\node at ( -3,0) [node] (src) {};
	\node at (3,0) [node] (dest) {};
	\draw [->,out=40,in=140] (src) to node [above] {$r_0^\mathrm{sc} $} (GW);
	\draw [dashed,->,out=0,in=180] (src) to node [above] {$r_{i}^\mathrm{sc} $} (GW);
	\draw [->,out=-40,in=-140] (src) to node [above] {$r_{n-1}^\mathrm{sc} $} (GW);
	\draw [->,out=40,in=140] (GW) to node [above] {$r_0^\mathrm{ca}$} (dest);
	\draw [dashed,->,out=0,in=180] (GW) to node [above] {$r_i^\mathrm{ca}$} (dest);
	\draw [->,out=-40,in=-140] (GW) to node [above] {$r_{n-1}^\mathrm{ca}$} (dest);
	\node [above] at (GW.north) {controller};
	\node [above] at (src.north) {sensor};
	\node [above] at (dest.north) {actuator};
\end{tikzpicture}
}
\caption{A flow with $ n $ sa-paths $ \pi_0, \dots, \pi_i, ..., \pi_{n-1} $, where $ r_i^\mathrm{sc} $ and $ r_i^\mathrm{ca} $ are the reliabilities of the sc-path and the ca-path of the sa-path $\pi_i$, respectively.}
\label{fig:n-path-flow}
\end{figure}
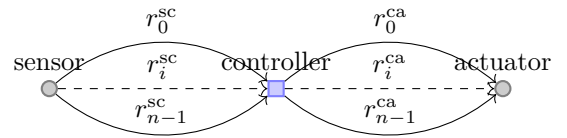
\begin{figure}[!h] 
   \centering
   \includegraphics[width=8cm]{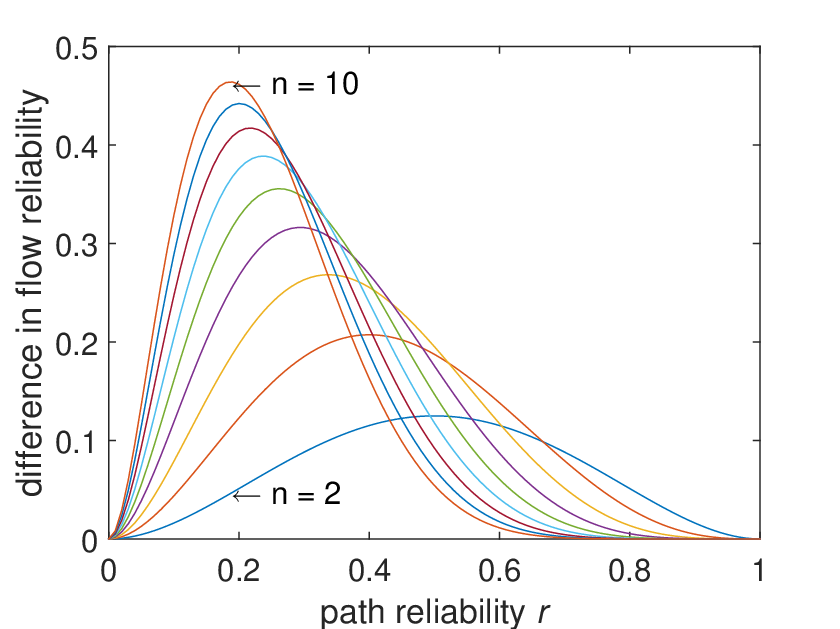}
   \caption{Difference in flow reliability, $r^\mathrm{2P} - r^\mathrm{1P}$, showing that \emph{2P scheduling} outperforms \emph{1P scheduling}.}
   \label{fig:difference-reliability}
\end{figure}

To illustrate the difference in flow reliability between \emph{1P scheduling} and \emph{2P scheduling}, we consider the special case in which all sc- and ca-paths have the same reliability $r$. We illustrate the difference $ r^\mathrm{2P} - r^\mathrm{1P} = [1 - (1 - r)^n]^2  - [1 - (1-r^2)^n] $ in \cref{fig:difference-reliability}, varying $ r \in \left[ 0, 1 \right] $ and $ n \in \{2 \dots 10\} $. We observe that as the number of paths increases, the difference in reliability changes. In particular, the difference increases for smaller values of $ r $, whereas it decreases for larger values of $ r $. Despite \emph{2P scheduling} being more reliable than \emph{1P scheduling}, the minimum delay of the former is larger, as shown next.

\begin{thm}\label{thm:delay} Given the same multi-path routing of a flow, the minimum end-to-end flow delay of 2P scheduling is greater than or equal to that of 1P scheduling, i.e., $ d^\mathrm{2P} \geq d^\mathrm{1P} $.
\end{thm}	

\begin{proof} 
The reader is referred to Appendix \ref{app:thm-2}.
\end{proof}

\emph{\textsc{Remark:} Although scheduling feasibility is related to both reliability (see \cref{thm:reliability}) and delay (see Theorem \ref{thm:delay}), reliability is more relevant when deadlines are not overly tight, especially in control applications such as our case. In particular, if the scheduling achieves a delay within the required deadline, it fulfills the control objectives without affecting performance.}

\subsection{Proposed LLF-RC Algorithm} \label{sec:llf-rc}
The \gls{LLF} algorithm was developed in \cite{liu2000:real-time-systems}, being its core idea the prioritization of jobs with smaller laxity. We propose the \gls{LLF-RC} algorithm, which extends \gls{LLF}'s idea by prioritizing flows based on the number of remaining conflicts when laxities are equal. In particular, the laxity $ a (\cdot) $ of a released transmission $ \tau $ is defined as its absolute deadline $d_\mathrm{abs}(\tau)$ minus the current time $t$, as shown below
\begin{align}
	a (\tau) = d_\mathrm{abs}(\tau) - t  \label{eqn:laxity}.
\end{align}

The absolute deadline of a released transmission $\tau$ is computed as the absolute deadline of the $k$-th activation of path $ d_\mathrm{abs}^{\mathrm{path}}(\pi_i^{x}, k) $ minus the number of hops ahead on that path, i.e., $ n_\mathrm{ahead}(\pi_i^{x}) $. Mathematically, we express it as
\begin{align} \nonumber
	 d_\mathrm{abs}(\tau) & = d_\mathrm{abs}^{\mathrm{path}}(\pi_i^{x}, k) - n_\mathrm{ahead}(\pi_i^{x}) \\ 
	                      & = k \cdot p_f + d_\mathrm{rel}^\mathrm{path}(\pi_i^{x}) - 1 - n_\mathrm{ahead}(\pi_i^{x})
	\label{eq:abs_dl_tx},  
\end{align}
where $ x \in \left\lbrace \mathrm{sc}, \mathrm{ca} \right\rbrace $ and $ p_f $ is the flow period. The relative path deadline $ d_\mathrm{rel}^\mathrm{path}(\pi_i^{x}) $ for sc- and ca-paths are defined below
\begin{align} 
	d_\mathrm{rel}^\mathrm{path}(\pi_i^\mathrm{sc}) & = d_\mathrm{rel}^\mathrm{flow} (f) - \underset{j}{\mathrm{max}} \{\mathrm{hops}(\pi_j^\mathrm{ca})\}, \label{eqn:relative-path-deadline-sc} \\ 
	d_\mathrm{rel}^\mathrm{path}(\pi_i^\mathrm{ca}) & = d_\mathrm{rel}^\mathrm{flow} (f). \label{eqn:relative-path-deadline-ca}
\end{align}

If two released transmissions have the same laxity, \gls{LLF-RC} assigns a higher priority to the transmission that has a larger number of remaining conflicting transmissions. The reason is that such links are likely to become bottlenecks, which may impair transmission parallelism \cite{yuan2012:tree-based-multi-channel-convergecast-wireless-sensor-networks}. We denote the number of remaining conflicting transmissions by $ n^\mathrm{rem}(\tau) $, and define it as the total number of unscheduled transmissions in the set of conflicting links of $ \tau $.
The set of conflicting links of $\tau$ is composed of all links that share at least one node with the link of $ \tau $, including the link of $ \tau $ itself.

To illustrate \gls{LLF-RC} algorithm, we provide an example in \cref{fig:llf-rc-system-example}. We assume a network with $ 13 $ motes and $ 2 $ gateways. In particular, the motes consist of $ 2 $ sensors denoted by $ s_0, s_1 $; $ 2 $ actuators denoted by $ a_0, a_1 $; and $ 9 $ relays denoted by $ r_0, \dots, r_8 $. The gateways are denoted by $ g_0, g_1 $. There are two scheduled flows, $ f_0 $ (from sensor $s_0$ to actuator $a_0$) and $f_1$ (from sensor $s_1$ to actuator $a_1$). Here, $f_0$ has two sa-paths ($\pi_{00} = \pi_{00}^\mathrm{sc} \cup \pi_{00}^\mathrm{ca}$ and $\pi_{01} = \pi_{01}^\mathrm{sc} \cup \pi_{01}^\mathrm{ca}$), whereas $f_1$ has one sa-path ($\pi_{10} = \pi_{10}^\mathrm{sc} \cup \pi_{10}^\mathrm{ca}$). The periods and deadlines of $f_0$ and $f_1$ are $ p_0 = 10, d_0 = 10 $ and $ p_1 = 20, d_1 = 9 $, respectively. Also, two channels are available, i.e., $ C = 2 $.
\begin{figure}[!h]
	\centering
	\begin{tikzpicture}
		\draw (0,0) grid +(4,2);
		\node at (0,0) [] (a1) {};
		\node at (1,0) [] (g) {};
		\node at (2,0) [] (h) {};
		\node at (3,0) [] (i) {};
		\node at (4,0) [] (a0) {};
		\node at (0,1) [] (d) {};
		\node at (1,1) [] (g0) {};
		\node at (2,1) [] (e) {};
		\node at (3,1) [] (g1) {};
		\node at (4,1) [] (f) {};
		\node at (0,2) [] (s0) {};
		\node at (1,2) [] (a) {};
		\node at (2,2) [] (b) {};
		\node at (3,2) [] (c) {};
		\node at (4,2) [] (s1) {};

		\node [above left] at (a) {$r_0$};
		\node [above left] at (b) {$r_1$};
		\node [above left] at (c) {$r_2$};
		\node [above left] at (s0) {$s_0$};
		\node [above left] at (s1) {$s_1$};
		\node [above left] at (d) {$r_3$};
		\node [above left] at (g0) {$g_0$};
		\node [above left] at (e) {$r_4$};
		\node [above left] at (g1) {$g_1$};
		\node [above left] at (f) {$r_5$};
		\node [above left] at (a1) {$a_1$};
		\node [above left] at (g) {$r_6$};
		\node [above left] at (h) {$r_7$};
		\node [above left] at (i) {$r_8$};
		\node [above left] at (a0) {$a_0$};

		\fill[color = black] (1,1) circle (2pt);
		\fill (3,1) circle (2pt);
		
		\draw[fill = white] (0, 0) circle (2pt);		
		\draw[fill = white] (1, 0) circle (2pt);
		\draw[fill = white] (2, 0) circle (2pt);
		\draw[fill = white] (3, 0) circle (2pt);
		\draw[fill = white] (4, 0) circle (2pt);
		\draw[fill = white] (0, 1) circle (2pt);
		\draw[fill = white] (2, 1) circle (2pt);
		\draw[fill = white] (4, 1) circle (2pt);
		\draw[fill = white] (0, 2) circle (2pt);
		\draw[fill = white] (1, 2) circle (2pt);
		\draw[fill = white] (2, 2) circle (2pt);
		\draw[fill = white] (3, 2) circle (2pt);
		\draw[fill = white] (4, 2) circle (2pt);

		\draw[->, rounded corners, color = blue] (-0.3, 2) -- (-0.3, 0.7) -- (1, 0.7); 
		\draw[->, rounded corners, densely dotted, color = brown] (0, 2.3) -- (3.3, 2.3) -- (3.3, 1); 
		\draw[->, rounded corners, color = blue] (1.3, 0.7) -- (1.3, -0.3) -- (4, -0.3); 
		\draw[->, rounded corners, densely dotted, color = brown] (3.3, 0.7) -- (4.3, 0.7) -- (4.3, 0); 
        \draw[->, rounded corners, dashed, color = magenta] (4, 2.2) -- (2.8, 2.2) -- (2.8, 1); 
        \draw[->, rounded corners, dashed, color = magenta] (0.8, 1) -- (0.8, -0.2) -- (0, -0.2);	

		\draw[->, color = blue] (4.5, 2) -- (5, 2) node [right] {\textcolor{black}{$f_0, \pi_{00}: \pi_{00}^\mathrm{sc}+ \pi_{00}^\mathrm{ca}$}};
		\draw[->, densely dotted, color = brown] (4.5, 1.5) -- (5, 1.5) node [right] {\textcolor{black}{$f_0, \pi_{01}: \pi_{01}^\mathrm{sc} + \pi_{01}^\mathrm{ca}$}};
		\draw[->, dashed, color = magenta] (4.5, 1) -- (5, 1) node [right] {\textcolor{black}{$f_1, \pi_{10}: \pi_{10}^\mathrm{sc} + \pi_{10}^\mathrm{ca}$}};
	\end{tikzpicture}
	\caption{Example of a scheduling problem.}
	\label{fig:llf-rc-system-example}
\end{figure}
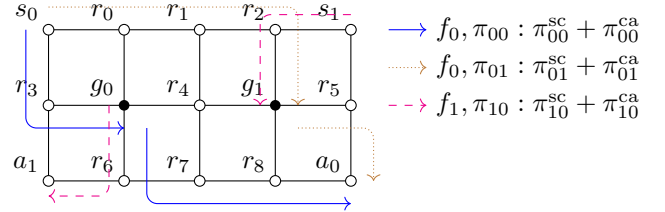

The scheduling performed by \gls{LLF-RC} algorithm is illustrated in Fig. \ref{fig:llf-rc-scheduling-example}. At the beginning of slot $ 3 $, $\pi_{00}^\mathrm{sc}$ has completed the transmissions, but $\pi_{01}^\mathrm{sc}$ has not. Note that the released transmission of $\pi_{01}^\mathrm{sc}$ is $\overrightarrow{r_2 g_1}$, which has absolute deadline $ 5 $. The released transmission of $\pi_{10}^\mathrm{sc}$ is also $\overrightarrow{r_2 g_1}$, but its absolute deadline is $ 6 $. The two transmissions have laxities $ 2 $ and $ 3 $, respectively, while the number of remaining conflicting transmissions for both is $ 6 $, equal to the sum of unscheduled transmissions on $ \overrightarrow{r_2 g_1} $, $ \overrightarrow{r_1 r_2} $, $ \overrightarrow{s_1 r_2} $, and $ \overrightarrow{g_1 r_5} $. As a result, $ \pi_{01}^\mathrm{sc} $ is scheduled due to its smaller laxity.
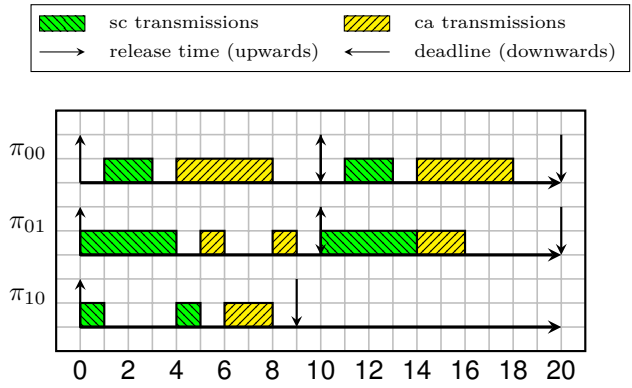
\begin{figure}[!h]
	\begin{tikzpicture}
		\centering
   		\begin{axis}
   		[%
   			hide axis,
		   	height = 1.6cm,
		   	xmin = 0,
		   	xmax = 1,
		   	ymin = 0,
		   	ymax = 1,
			legend style = {row sep = 0.01cm},
		 	legend style = {column sep = 0.25cm},
		   	legend cell align = {left},
		   	legend columns = 2,
		   	legend pos = north east,
		   	legend style = {at = {(3.1, 4)}, anchor = south west, font = \fontsize{7}{6}\selectfont, text depth = .ex, fill = none},
	  	]
	
	 	\addlegendimage{draw=none, fill = green, area legend, postaction = {pattern = north west lines, pattern color = black}} 	\addlegendentry{sc transmissions}

	 	\addlegendimage{draw=none, fill = yellow, area legend, postaction = {pattern = north east lines, pattern color = black}} 	\addlegendentry{ca transmissions}
	 	
	 	\addlegendimage{-stealth } \addlegendentry{release time (upwards)}
	 	
		\addlegendimage{stealth-} \addlegendentry{deadline (downwards)}
	   
	   \end{axis}

	\end{tikzpicture}
	\centering
	\begin{RTGrid}[nosymbols=1,width=7cm]{3}{20}
	\RowLabel{1}{$\pi_{00}$}
	\RowLabel{2}{$\pi_{01}$}
	\RowLabel{3}{$\pi_{10}$}
	\TaskArrDead{1}{0}{10}
	\TaskArrDead{1}{10}{10}
	\TaskArrDead{2}{0}{10}
	\TaskArrDead{2}{10}{10}
	\TaskArrDead{3}{0}{9}
	\TaskExecutionSC[color = green]{3}{0}{1}
	\TaskExecutionCA[color = yellow]{2}{5}{6}
	\TaskExecutionSC[color = green]{2}{0}{4}
	\TaskExecutionSC[color = green]{1}{1}{3}
	\TaskExecutionCA[color = yellow]{1}{4}{8}
	\TaskExecutionCA[color = yellow]{2}{8}{9}
	\TaskExecutionSC[color = green]{3}{4}{5}
	\TaskExecutionCA[color = yellow]{3}{6}{8}
	\TaskExecutionCA[color = yellow]{2}{14}{16}
	\TaskExecutionSC[color = green]{2}{10}{14}
	\TaskExecutionCA[color = yellow]{1}{14}{18}
	\TaskExecutionSC[color = green]{1}{11}{13}
	\end{RTGrid}
	\caption{Scheduling with LLF-RC}
	\label{fig:llf-rc-scheduling-example}
\end{figure}

\subsection{Opportunistic Aggregation} \label{sec:opportunistic-aggregation}

We propose \gls{OA} mechanism, which can be integrated to work seamlessly with any scheduling algorithm. In general, scheduling problems can be infeasible due to tight deadlines or limitations on the total utilization. For instance, when the total utilization is greater than the number of channels $ C $, scheduling problems are infeasible. Even if there are enough available channels, scheduling can still be infeasible due to strict deadlines. \Gls{OA} can alleviate the schedulability issue in either case by combining packets. Specifically, \gls{OA} can be applied to most control applications as control packets are shorter compared to the slot length, which is long enough to accommodate a packet of maximum size, as defined in WirelessHART \cite{HART07} and \gls{TSCH} \cite{tsch12} standards. To leverage \gls{OA}, we replace \texttt{line 6} in \cref{alg:framework} with \cref{alg:opportunistic-aggregation}, which determines the packets to be scheduled and aggregated. The logic of \cref{alg:opportunistic-aggregation} is as follows. When the sender of a released transmission is already scheduled but the receiver is not or the link itself is scheduled, the packets are aggregated (see \texttt{lines 3-6}). When the sender and the receiver are both unscheduled, and there is at least one unused channel, then an unused channel is employed for the transmission (see \texttt{lines 8-12}).


{\fontsize{11}{10}\selectfont
\begin{algorithm}
	\KwIn{ \\ $ C: $ number of channels \\ $ \mathcal{J}: $ set of released transmissions ordered by priority }
	\KwOut{ \\ $ \mathcal{S}: $ set of scheduled senders \\$ \mathcal{R}: $ set of scheduled receivers \\$ \mathcal{L}: $ set of scheduled links}
	\KwInit{$ m = 0, \mathcal{S} = \emptyset, \mathcal{R} = \emptyset, \mathcal{L}  = \emptyset$\;}
	
	\ForEach{$ t \in \mathcal{J} $} {
	\eIf {$ t.\mathrm{sender} \in \mathcal{S} $} {
			\If {($ t.\mathrm{receiver} \notin \mathcal{S} $ and $ t.\mathrm{receiver} \notin  \mathcal{R} $) or $(t.\mathrm{sender} \rightarrow t.\mathrm{receiver}) \in \mathcal{L} $} {
				Schedule $ t $ and perform aggregation\;
				$ \mathcal{R} = \mathcal{R} \cup t.\mathrm{receiver} $\; 
				$ \mathcal{L} = \mathcal{L} \cup (t.\mathrm{sender} \rightarrow t.\mathrm{receiver}) $\;
			} 
		} {
			\If {$ m < C $ and $ t.\mathrm{sender} \notin \mathcal{R} $ and $t.\mathrm{receiver} \notin \mathcal{S} \cup \mathcal{R} $ } {
				Schedule $t$\;
				$ \mathcal{S} = \mathcal{S} \cup t.\mathrm{sender}, \mathcal{R} = \mathcal{R} \cup t.\mathrm{receiver} $\; 
				$ \mathcal{L} = \mathcal{L} \cup (t.\mathrm{sender} \rightarrow t.\mathrm{receiver}) $\;
				$ m = m +1 $\;
			}
		}
	}	
	\caption{Proposed OA mechanism}
	\label{alg:opportunistic-aggregation}
	\algorithmfootnote{Note that $ t.\mathrm{sender} $ denotes the sender of released transmission $ t $, and $ t.\mathrm{receiver} $ denotes the receiver of released transmission $ t $. To use \gls{OA}, we disable the utilization check in the schedulability test. }
\end{algorithm}
}

To illustrate \gls{OA}'s principle, we use the same example as in \cref{fig:llf-rc-system-example} with \gls{LLF-RC} algorithm, and show the resulting scheduling in \cref{fig:llf-rc-oa-scheduling-example}. We observe that \gls{OA} takes effect in slots $ 0 $, $ 3 $, $ 4 $, $ 5 $, and $ 10 $. In slot $ 0 $, the same packet is sent on both links $ \overrightarrow{s_0 r_0} $ and  $ \overrightarrow{s_0 r_1} $, simultaneously. In slot $ 3 $, two packets from different flows are aggregated and sent on link $\overrightarrow{r_2 g_1}$. In slot $ 5 $, two packets from different flows are aggregated on node $ r_6 $ and sent on links $ \overrightarrow{r_6 a_1} $ and $ \overrightarrow{r_6 r_7} $.
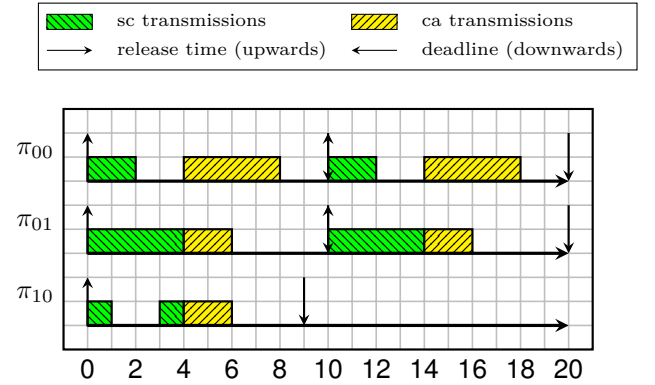
\begin{figure}[!h]
	\begin{tikzpicture}
		\centering
   		\begin{axis}
   		[%
   			hide axis,
		   	height = 1.6cm,
		   	xmin = 0,
		   	xmax = 1,
		   	ymin = 0,
		   	ymax = 1,
			legend style = {row sep = 0.01cm},
		 	legend style = {column sep = 0.25cm},
		   	legend cell align = {left},
		   	legend columns = 2,
		   	legend pos = north east,
		   	legend style = {at = {(3.1, 4)}, anchor = south west, font = \fontsize{7}{6}\selectfont, text depth = .ex, fill = none},
	  	]
	
	 	\addlegendimage{draw=none, fill = green, area legend, postaction = {pattern = north west lines, pattern color = black}} 	\addlegendentry{sc transmissions}

	 	\addlegendimage{draw=none, fill = yellow, area legend, postaction = {pattern = north east lines, pattern color = black}} 	\addlegendentry{ca transmissions}
	 	
	 	\addlegendimage{-stealth } \addlegendentry{release time (upwards)}
	 	
		\addlegendimage{stealth-} \addlegendentry{deadline (downwards)}
	   
	   \end{axis}

	\end{tikzpicture}
	\centering
	\begin{RTGrid}[nosymbols=1,width=7cm]{3}{20}
	\RowLabel{1}{$\pi_{00}$}
	\RowLabel{2}{$\pi_{01}$}
	\RowLabel{3}{$\pi_{10}$}
	\TaskArrDead{1}{0}{10}
	\TaskArrDead{1}{10}{10}
	\TaskArrDead{2}{0}{10}
	\TaskArrDead{2}{10}{10}
	\TaskArrDead{3}{0}{9}
	\TaskExecutionSC[color = green]{3}{0}{1}
	\TaskExecutionSC[color = green]{2}{0}{4}
	\TaskExecutionSC[color = green]{1}{0}{2}
	\TaskExecutionCA[color = yellow]{1}{4}{8}
	\TaskExecutionCA[color = yellow]{2}{4}{6}
	\TaskExecutionSC[color = green]{3}{3}{4}
	\TaskExecutionSC[color = green]{2}{10}{14}
	\TaskExecutionCA[color = yellow]{1}{14}{18}
	\TaskExecutionCA[color = yellow]{3}{4}{6}
	\TaskExecutionSC[color = green]{1}{10}{12}
	\TaskExecutionCA[color = yellow]{2}{14}{16}
	\end{RTGrid}
	\caption{Scheduling with LLF-RC and OA}
	\label{fig:llf-rc-oa-scheduling-example}
\end{figure}

\subsection{Repetitive Scheduling} \label{sec:repetitive-scheduling}
Centralized \gls{TDMA} usually requires a network manager to (i) compute the scheduling, and (ii) distribute it to each mote. Since each node needs to store at least the partial scheduling related to itself, the communication cost (to disseminate the scheduling) and the storage cost (to retain the scheduling) cannot be ignored. Assuming that the communication and storage costs for every entry of the scheduling table is constant, the total cost of either is $ \mathcal{O} (H \cdot U) $, where $ H $ is the hyperperiod, $ U $ is the total utilization, and $ H \cdot U $ is the total number of entries in the scheduling table. In addition, $ H $ can be as large as the product of all periods, i.e., $ H = \prod_i {p_i} $, which occurs when the periods are co-prime. However, if the periods are \emph{harmonic}, meaning a period is divisible by any period smaller than itself, then $ H $ is equal to the maximum period, i.e., $H = \max_i \left\lbrace p_i \right\rbrace $, thus leading to a smaller cost.

If a system has harmonic periods, communication and storage costs can be reduced by enforcing that scheduling of each flow is repeated for each flow period. Therefore, we call this strategy \gls{RS}. In contrast, \gls{HS} is the procedure by which all time slots of a hyperperiod are scheduled. Both \gls{RS} and \gls{HS} behave similarly in terms of execution time. However, \gls{HS} has a much larger communication and storage cost $\mathcal O(H \cdot U) = \mathcal O (\sum_i \frac{H}{p_i} \cdot L_i)$, whereas \gls{RS} has a more affordable cost $\mathcal O(\sum_i L_i)$, where $ L_i $ is the total number of hops in a flow, and is independent of $H$. The trade-off between these two types of scheduling is that \gls{RS} is less costly but also less flexible, which can affect scheduling performance, while \gls{HS} is more costly but more flexible, and generally offers higher schedulability.

In \cref{alg:repetitive-scheduling}, we describe the proposed \gls{RS}. In particular, \gls{RS} uses as input the scheduling result of a \gls{HS} algorithm, for instance, \gls{LLF-RC} presented in \cref{sec:llf-rc}. Then, it finds the periods of all flows (see \texttt{line 2}). Subsequently, it identifies the distinct values for the periods, and sort them in increasing order (see \texttt{line 3}). Afterwards, it schedules the flows of a certain period (see \texttt{line 5}). If the scheduling of a flow meets the deadline for one period, it will always do so for the entire hyperperiod since the periods are harmonic. Note that \gls{RS} can be combined with \gls{OA}. 

{\fontsize{11}{10}\selectfont
\begin{algorithm}
	\KwIn{ \\ $ \mathcal{A} $: A \gls{HS} algorithm }
	\KwOut{\\ $ \mathcal{H} $: Scheduling table, infeasible if $\mathcal{H} = \emptyset$}
	$\mathcal{H} = \emptyset$\;
	Find the periods of all $ F $ flows, listed as $ p_0, p_1 \dots p_{F-1}$\; 
	Find distinct values from $p_i $, and sort them, such that $p'_0 < p'_1 <...<p'_{M-1}, M \le F$\;
	\For {$j = 0$ \KwTo $M-1$} {
		Schedule all flows having the same period $p'_j$ for one period using $ \mathcal{A} $, considering the current scheduling table $ \mathcal{H}$; let the resulting scheduling table be $\mathcal{H}_j$\;
		\If {$\mathcal{H}_j = \emptyset$} {
			\Return $ \mathcal{H} = \emptyset $;	
		}
		$ \mathcal{H} = \mathcal{H} \cup \{ (p'_j, \mathcal{H}_j) \}$;	
	}
	\Return $ \mathcal{H}$;	
\caption{Proposed RS mechanism}
\label{alg:repetitive-scheduling}
\end{algorithm}
}

Considering the example in \cref{fig:llf-rc-system-example}, the scheduling that results from using \gls{LLF-RC} with \gls{RS} is shown in \cref{fig:llf-rc-rs-examples}. Contrary to what occurs in \cref{fig:llf-rc-scheduling-example}, where \gls{LLF-RC} is employed without \gls{RS}, the scheduling of $ f_0 $ is repeated throughout the periods.
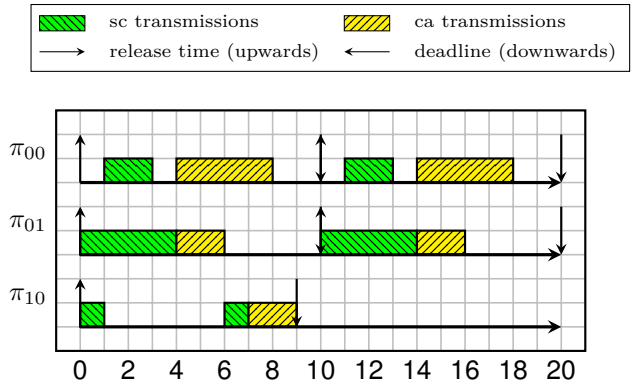
\begin{figure}[!h]
	\begin{tikzpicture}
		\centering
   		\begin{axis}
   		[%
   			hide axis,
		   	height = 1.6cm,
		   	xmin = 0,
		   	xmax = 1,
		   	ymin = 0,
		   	ymax = 1,
			legend style = {row sep = 0.01cm},
		 	legend style = {column sep = 0.25cm},
		   	legend cell align = {left},
		   	legend columns = 2,
		   	legend pos = north east,
		   	legend style = {at = {(3.1, 4)}, anchor = south west, font = \fontsize{7}{6}\selectfont, text depth = .ex, fill = none},
	  	]
	
	 	\addlegendimage{draw=none, fill = green, area legend, postaction = {pattern = north west lines, pattern color = black}} 	\addlegendentry{sc transmissions}

	 	\addlegendimage{draw=none, fill = yellow, area legend, postaction = {pattern = north east lines, pattern color = black}} 	\addlegendentry{ca transmissions}
	 	
	 	\addlegendimage{-stealth } \addlegendentry{release time (upwards)}
	 	
		\addlegendimage{stealth-} \addlegendentry{deadline (downwards)}
	   
	   \end{axis}

	\end{tikzpicture}
	\centering
	\begin{RTGrid}[nosymbols=1,width=7cm]{3}{20}
	\RowLabel{1}{$\pi_{00}$}
	\RowLabel{2}{$\pi_{01}$}
	\RowLabel{3}{$\pi_{10}$}
	\TaskArrDead{1}{0}{10}
	\TaskArrDead{1}{10}{10}
	\TaskArrDead{2}{0}{10}
	\TaskArrDead{2}{10}{10}
	\TaskArrDead{3}{0}{9}
	\TaskExecutionSC[color = green]{3}{0}{1}
	\TaskExecutionSC[color = green]{2}{0}{4}
	\TaskExecutionSC[color = green]{1}{1}{3}
	\TaskExecutionCA[color = yellow]{1}{4}{8}
	\TaskExecutionCA[color = yellow]{2}{4}{6}
	\TaskExecutionSC[color = green]{3}{6}{7}
	\TaskExecutionSC[color = green]{2}{10}{14}
	\TaskExecutionCA[color = yellow]{1}{14}{18}
	\TaskExecutionCA[color = yellow]{3}{7}{9}
	\TaskExecutionSC[color = green]{1}{11}{13}
	\TaskExecutionCA[color = yellow]{2}{14}{16}
	\end{RTGrid}
	\caption{Scheduling with LLF-RC and RS}
	\label{fig:llf-rc-rs-examples}
\end{figure}

\subsection{Implications of Applying the Scheduling Framework to Industrial Networks}
The following practical implications are essential for applying the proposed scheduling framework to industrial networks.

\textbf{Routing.} Routing is the prerequisite for applying the scheduling framework. The central unit of a network computes a multi-path routing scheme periodically (c.f. Section~\ref{sec:routing-model} and \ref{sec:generation-scenarios}). The routing algorithm requires as input the link quality measurement of the network. A mechanism that measures link quality locally, and convergecasts the measurements to the central unit is needed~\cite{yuan2012:tree-based-multi-channel-convergecast-wireless-sensor-networks}. The details are out of the scope of this paper.

\textbf{Scheduling.} After routing is done, the central unit is ready to run scheduling algorithms, e.g. \gls{LLF-RC} + \gls{RS} + \gls{OA}. Then the computed scheduling table needs to be distributed to each node in the network, which stores a relevant partial scheduling table. We show in Section~\ref{sec:perf-llc-rc-hs-rs} that the distribution of scheduling table takes less than one second for a network of 100 nodes and 50 flows with the Ripple protocol~\cite{yuan2015:ripple-high-throughput-reliable-energy-efficient-network-flooding-wireless-sensor-networks}.

\textbf{Periodic Control.} Now control loops are formed, and should be activated periodically according to the scheduling table. To further increase control quality, the control algorithm may compute the actuating packet even if the related sensing packet is not delivered to any of the gateways in time, by applying an control algorithm that can estimate the sensing data. 
Some intelligence may even be applied at an actuator which can tolerate sporadic sensing packet and actuating packet loss. The details are out of the scope of this paper.
\color{black}

%% file: sections/proposed-method.tex

\section{Adaptation of Existing Scheduling Algorithms} \label{sec:extended-framework}

In this section we adapt several existing scheduling algorithms to be compatible with our framework. 

\subsection{Existing Scheduling Algorithms}

In \cite{liu2000:real-time-systems}, various multiprocessor scheduling algorithms were proposed, namely, \gls{RM}, \gls{DM}, \gls{PDM}, \gls{CLLF}, \gls{EDF}, and \gls{EPD}. In addition, in \cite{lee1994:on-line-multiprocessor-scheduling-algorithms-real-time-tasks}, the \gls{EDZL} algorithm was proposed. All these multiprocessor scheduling algorithms can be classified into two categories, \emph{fixed-priority} and \emph{dynamic-priority}. We adapt these algorithms to our case by introducing the following changes.

\begin{itemize} [leftmargin = 0.3cm]
	\item We use the \emph{2P scheduling} instead of \emph{1P scheduling}, which is used by \cite{saifullah2010:real-time-scheduling-wirelesshart-networks, lee1994:on-line-multiprocessor-scheduling-algorithms-real-time-tasks}, etc.
	
	\item The transmissions of a sa-path is mapped to a \emph{task}. However, in our work, these transmissions are mapped to two \emph{tasks}, which correspond to the transmissions of sc- and ca-paths. 
\end{itemize}

In the following, we briefly revisit the aforementioned algorithms, adapting the terminology to our context.

\emph{\textsc{Remark:} Following the terminology of real-time scheduling, we use the term `task' to denote a repeated work to be scheduled, and the term `job' to denote an instance of a task.}

\subsection{Fixed-Priority Scheduling Algorithms} \label{sec:fixed-priority}

Fixed-priority scheduling algorithms assign priorities to tasks in advance. Hence, runtime information is not required for comparing flow priorities. In general, the performance of fixed-priority algorithms is inferior to that of dynamic-priority algorithms \cite{davis2011:survey-hard-real-time-scheduling-multiprocessor-systems}. 

\textbf{Rate monotonic (RM).} To be able to use \gls{RM} for flow scheduling, we map the transmissions of a flow to a task. Specifically, \gls{RM} assigns each flow a priority inversely proportional to its period \cite{liu2000:real-time-systems}.

\textbf{Deadline monotonic (DM).}  Analogous to \gls{RM}, we map the transmissions of a flow to a task. However, \gls{DM} assigns each flow a priority inversely proportional to its \emph{relative deadline} \cite{liu2000:real-time-systems}. The \emph{relative deadline} of a flow is denoted by $ d^\mathrm{flow}_\mathrm{rel} \left( f \right) $, and is determined by the control system design.

\textbf{Proportional deadline monotonic (PDM).} Compared to \gls{RM} and \gls{DM}, the tasks in \gls{PDM} have finer granularity. In \gls{PDM}, we map transmissions of a sc- or ca-path to a task. \gls{PDM} assigns each sc- or ca-path a priority inversely proportional to its \emph{proportional deadline} $ d_\mathrm{prop}^\mathrm{path}(\pi_i^{x}) $, which is calculated as,
\begin{equation}
	d_\mathrm{prop}^\mathrm{path}(\pi_i^{x}) = \cfrac{d^\mathrm{subflow} \left( \pi_i^{x} \right)}{l(\pi_i^{x})},
\end{equation}
where $ d^\mathrm{subflow} \left( \pi_i^{x} \right) $ is the \emph{subflow deadline}, and $ l(\pi_i^{x}) $ is the \emph{path length}, i.e., the number of hops of $ \pi_i^{x} $.
The subflow deadline $ d^\mathrm{subflow} \left( \pi_i^{x} \right) $ of an sc- or ca-path is defined as the relative flow deadline minus the path length of the longest ca- or sc-path, as shown below,
\begin{equation}
	d^\mathrm{subflow} \left( \pi_i^{x} \right) = d^\mathrm{flow}_\mathrm{rel} \left( f \right) - \max_j \{l(\pi_j^{y}) \},  
\end{equation}
where $ (x,y) \in \left\lbrace (\mathrm{sc}, \mathrm{ca}), (\mathrm{ca}, \mathrm{sc}) \right\rbrace $.

\subsection{Dynamic-Priority Scheduling Algorithms} \label{sec:dynamic-priority}

Compared to fixed-priority algorithms, which employ a predefined priority for the whole scheduling process, dynamic-priority algorithms determine the flows priority at runtime.

\textbf{Earliest deadline first (EDF).} It prioritizes the jobs with earlier absolute deadlines \cite{liu2000:real-time-systems}. We map the transmissions of a sc- or ca-path in one period to a job. Assuming that a flow $ f $ has sc-paths $ \pi_i^\mathrm{sc} $, ca-paths $ \pi_i^\mathrm{ca} $, and a \emph{relative flow deadline} of $ d_\mathrm{rel}^\mathrm{flow} \left( f \right) $, the relative deadlines of a sc- and ca-path are $ d_\mathrm{rel}^\mathrm{path} \left( \pi_i^\mathrm{sc} \right) $ and $ d_\mathrm{rel}^\mathrm{path} \left( \pi_i^\mathrm{ca} \right) $, respectively, which were defined in (\ref{eqn:relative-path-deadline-sc}) and (\ref{eqn:relative-path-deadline-ca}). Thus, the absolute deadline of the $k$-th activation of an sc- or ca-path (namely, the absolute deadline of the transmissions of the $k$-th packet of flow $f$ on the path) is calculated as
\begin{align}
	d_\mathrm{abs}^\mathrm{path}(\pi_i^{x}, k) = k \cdot p_f + d_\mathrm{rel}^\mathrm{path}(\pi_i^{x}) - 1
	\label{eqn:absolute-deadline}, x \in \{\mathrm{sc}, \mathrm{ca}\}, 
\end{align}
where $ p_f $ is the flow period. Here, $ d_\mathrm{abs}^\mathrm{path}(\pi_i^{x}, k) $ is used to dynamically assign priorities to jobs.

\textbf{Earliest proportional deadline (EPD).} It prioritizes jobs with small sub-deadlines. A sub-deadline is defined as the time left until the deadline of a job, divided by the processing cost needed to finish the job. In our case, the transmissions in a sc- or ca-path are mapped to a job. Therefore, we define a sub-deadline as the remaining slots left until the path deadline, divided by the number of remaining transmissions on the path, which is given by
\begin{align}
	s^\mathrm{path}(\pi_i^{x}, k) & = \frac{d_\mathrm{abs}^\mathrm{path}(\pi_i^{x}, k) - t +1}{n_\mathrm{rem}(\pi_i^{x})} \nonumber \\
								  & = \frac{d_\mathrm{abs}^\mathrm{path}(\pi_i^{x}, k) - t +1}{n_\mathrm{ahead}(\pi_i^{x}) + 1}, x \in \{\mathrm{sc}, \mathrm{ca}\}, 
\end{align}
where $k$ is the current period index, $ t $ is the current slot index, and $ n_\mathrm{rem}(\pi_i^{x}) $ is the number of remaining transmissions. Recall that $ n_\mathrm{ahead}(\pi_i^{x}) $ was defined in \cref{sec:llf-rc}.

\textbf{Earliest deadline zero laxity (EDZL).} It yields the same scheduling as \gls{EDF} until the instant when a job has a laxity of zero \cite{lee1994:on-line-multiprocessor-scheduling-algorithms-real-time-tasks}. At this instant, \gls{EDZL} gives the job the highest priority. We implement \gls{EDZL} as a combination of \gls{EDF} and \gls{LLF}. If two jobs have laxity greater than $ 0 $ and unequal deadlines, \gls{EDZL} works as \gls{EDF}, otherwise, it works as \gls{LLF}.

\textbf{Conflict-aware least-laxity-first (CLLF).} It is an extension of the \gls{LLF} algorithm adapted to WirelessHART by taking into account conflicts among remaining transmissions. Since it needs to search for unreleased transmissions, it is computationally expensive. The authors report that its performance is significantly higher than other scheduling algorithms evaluated in \cite{saifullah2010:real-time-scheduling-wirelesshart-networks}. To be compatible with \emph{2P scheduling}, we adapted \gls{CLLF} by properly calculating the deadline of a transmission according to (\ref{eqn:absolute-deadline}).

%% file: sections/simulation-results.tex
\section{Simulations} \label{sec:simulations}
We evaluate the performance of several scheduling algorithms using randomly generated networks based on the model described in  \cref{sec:link-quality-model}. We also evaluate the impact of \gls{OA} and \gls{RS}~\footnote{The source code and dataset for
the simulations in this paper are available at: https://github.com/pdvyuan/WSAN\_periodic\_control\_scheduling}.

\subsection{Performance Metrics}

To evaluate the algorithms' performance, we use the following metrics:

\textbf{Schedulability ratio.} It is related to schedulability performance. It denotes the percentage of instances in which all deadlines are met. A higher value is desirable.

\textbf{Execution time.} It is the time an algorithm needs to return a scheduling solution. A smaller value is desirable.

\textbf{Maximum queue length.} It is related storage costs. It is the maximum number of packets stored by a node during network operation. A smaller value is desirable.

\textbf{Scheduling table size.} It is related to communication and storage costs. It is the size of the table that contains a schedule. A smaller value is desirable.

\subsection{Generation of Scenarios} \label{sec:generation-scenarios}

In order to evaluate a large variety of scenarios, we examine different settings of topology, number of flows, total utilization, period length, deadline type, and number of channels, as explained next. We summarize the most relevant parameters in  \cref{tab:simulation-parameters}.
\begin{table}[!h]
	\renewcommand{\arraystretch}{1.15}
	\setlength\tabcolsep{6pt} 
	\caption{Simulation parameters}
	\centering
	\fontsize{7pt}{9pt}\selectfont
	\label{tab:simulation-parameters}
	\begin{tabular}{|p{2.8cm}| p{1.3cm} | p{3.5cm} |}
		\hline
		\textbf{ Description } & \textbf{ Parameter }  &	\textbf{ Value and Units} 		\\
		\hline
		\hline
		Side of the square & $ s $ & 1200 m	\\ \hline
		Number of motes & $ N $ & 100	\\ \hline
		Transmit power & $ P_\mathrm{tx} $ & 0 dBm \\ \hline
		Packet length & $ L $ & 133 bytes \\ \hline
		Number of flows & $ F $ & $ \left\lbrace 1, 2, \dots, F_\mathrm{max} \right\rbrace $ \\ \hline
		Maximum number of flows & $ F_\mathrm{max} $ & $ 50 $ \\ \hline
		Expected total utilization & $ U^\mathrm{exp} $ & $ [0, U_\mathrm{max}] $ \\ \hline
		Maximum total utilization & $ U_\mathrm{max} $ & $ 25 $ (with \gls{OA}) or  $ 16 $ (without \gls{OA}) \\ \hline
		Number of channels & $ C $ & $ \{1, \dots, 16 \}$ \\ \hline
		Number of topologies & $ N_\mathrm{top} $ & $ 100 $ \\ \hline
		Parameter (TOSSIM model) & $ \beta_1 $ & $ 0.9794 $ \\ \hline
		Parameter (TOSSIM model) & $ \beta_2 $ & $ 2.3851 $ \\ \hline
		PPR threshold & $ \Gamma_\mathrm{th} $ & $ 0.5 $ \\ \hline
	\end{tabular}
\end{table}

\begin{figure}[!t] 
	\centering
	\begin{subfigure}{0.9\columnwidth}
	   	\centering
		\includegraphics[width=\textwidth]{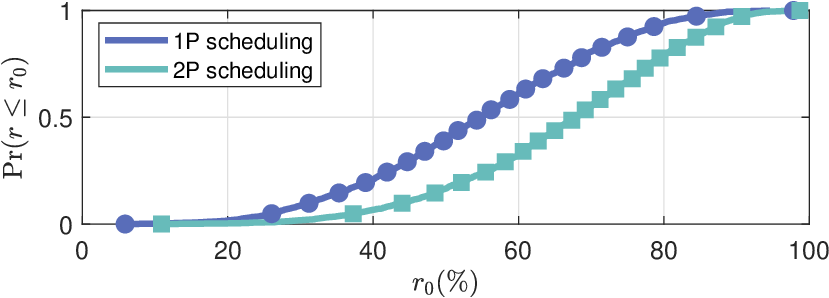}
		\caption{Link quality $ [20, 100] \% $}
	   	\label{fig:comparison-reliability-1p-2p-scheduling-02}
	   	\vspace{2mm}
	\end{subfigure}
	\begin{subfigure}{0.9\columnwidth}
		\centering
		\includegraphics[width=\textwidth]{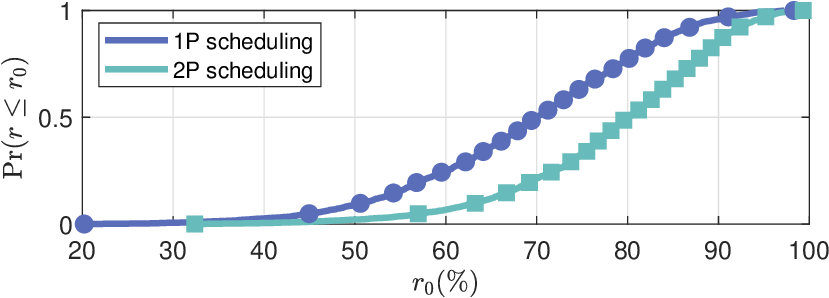}
		\caption{Link quality $ [40, 100] \% $}
	   	\label{fig:comparison-reliability-1p-2p-scheduling-04}
	   	\vspace{2mm}
	\end{subfigure}
	\begin{subfigure}{0.9\columnwidth}
	    \centering
		\includegraphics[width=\textwidth]{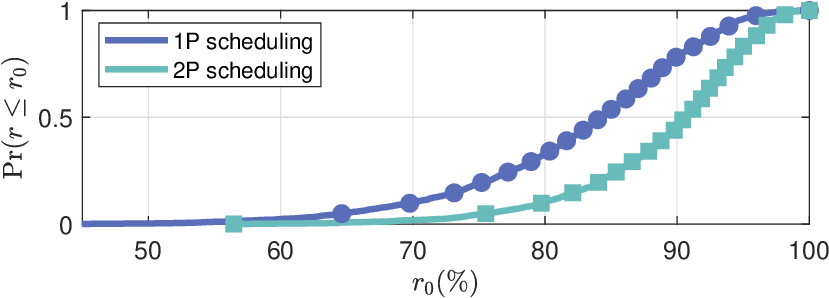}
		\caption{Link quality $ [60, 100] \% $}
	    \label{fig:comparison-reliability-1p-2p-scheduling-06}
	\end{subfigure}
	\begin{subfigure}{0.9\columnwidth}
	    \centering
		\includegraphics[width=\textwidth]{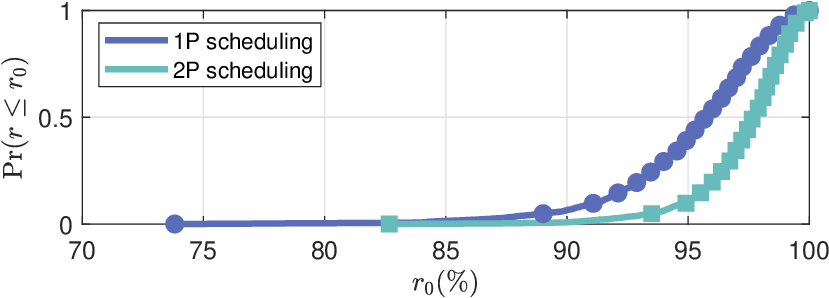}
		\caption{Link quality $ [80, 100] \% $}
	    \label{fig:comparison-reliability-1p-2p-scheduling-08}
	\end{subfigure}
	\caption{Empirical CDFs of communication reliability of \emph{1P scheduling} and \emph{2P scheduling}.}
	\label{fig:comparison-reliability-1p-2p-scheduling}
\end{figure}

\textbf{Topology.} We randomly generate each topology\footnote{The generated networks are dense with an average node degree of approximately $ 10 $.} with $ N $ motes uniformly distributed within a square space of size $s \times s$, which corresponds to medium-sized networked control systems. In addition, we place two gateways at the centers of the left and right half-planes of the square space. The link qualities are calculated with the model presented in \cref{sec:link-quality-model}. In total, $ N_\mathrm{top} $ topologies are evaluated.

\textbf{Number of flows.} We generate a random number of flows $ F \in \left\lbrace 1 \dots F_\mathrm{max} \right\rbrace  $, and for each topology, we evaluate $ 5 $ different flow settings. We enforce each flow to have at least two vertex-disjoint paths. 
	
\textbf{Total utilization.} We randomly choose an expected total utilization $ U^\mathrm{exp} \in [0, U_\mathrm{max}] $ for the whole network, where $ U_\mathrm{max} $ denotes the maximum total utilization value. We use the \emph{UUniFast} algorithm \cite{bini2005:measuring-performance-schedulability-tests,davis2011:improved-priority-assignment-global-fixed-priority-pre-emptive-scheduling-multiprocessor-real-time-systems}  to uniformly distribute $ U^\mathrm{exp} $ among all flows, and for each configuration of flows, we evaluate $ 10 $ different values of $ U^\mathrm{exp} $. Let flow $f$ have an expected utilization $ u_f^\mathrm{exp} $, and contain $\mathrm{hops}(f)$ links, then the period of $f$ is set to the minimum integer $p_f \ge \lceil \mathrm{hops}(f) / u_f^\mathrm{exp} \rceil$, which satisfies the constraint of period length. Thus, the actual utilization of the flow is $ u_f = \mathrm{hops}(f) / p_f $. 

\textbf{Period length.} To compare \gls{RS} and \gls{HS}, we use harmonic periods $ p_f = 2^k, k \in \{1 \dots 13\}$.
Otherwise, $p_f$ is a factor of $ 10000 $, meaning $ 10000 $ is divisible by the positive integer $p_f$. 

\textbf{Deadline type.} We choose between implicit and restricted deadline. The former implies that the deadline and the period are equal, and the latter one implies that the period is larger than the deadline. 

\textbf{Number of channels.} In IEEE 802.15.4, the maximum number of orthogonal channels is $ 16 $. Therefore, we consider $ C = \left\lbrace 1, 2, 4, 8, 16 \right\rbrace $. 

\textbf{Path-loss model.} We choose $ d_0 = 15 $ m, $ \mathrm{PL}(d_0) = 71.84 $ dBm, $ \eta = 2.16 $, and $ \sigma = 8.13 $, according to \cite{tanghe2008:industrial-indoor-channel-large-scale-temporal-fading-900-2400-5200-mhz}. These are widely adopted values for 2.4 GHz frequency in indoor factory environments, for both \gls{LOS} and \gls{NLOS}. A typical IEEE 802.15.4 radio chip CC2420 \cite{CC2420} has a transmit power $ P_\mathrm{tx} $ in the range $ \left[ -25, 0 \right] $ dBm. We choose $ P_\mathrm{tx} = 0 $ dBm since higher transmit power yields higher packet reliability when spatial reuse is disabled. Also, the noise floor of the CC2420 radio chip is $ -98 $ dBm.

\textbf{Communication model.} We choose maximum packet length $ L = 133 $ bytes to account for the worst case. Specifically, an IEEE 802.15.4 frame contains a preamble of $ 4 $ bytes, $ 1 $ SFD byte, $ 1 $ byte of frame length, and PSU of variable length up to $ 127 $ bytes. 

\textbf{Routing model.} We use \emph{two vertex-disjoint reliable paths} for every flow. First, the sc-routing is performed by finding the most reliable path. Then, we remove all nodes on the path except the sensor, and again find the most reliable path from the sensor to the other gateway. These two routes are disjoint except at the sensor node. Next, we restore the original topology, and find two disjoint ca-routes in the same way. The routing method does not cause any single point of failure because the failure of a node (except for a sensor or actuator) will not break the connectivity of a flow.

\textbf{Software.} The simulation is implemented in Java, and runs on a computer with an Intel Core i$7@2.6$ GHz CPU and $ 16 $ GB RAM.

\emph{\textsc{Remark:} In the following sections, we evaluate the performance of the proposed \gls{LLF-RC} algorithm and the adapted \gls{RM}, \gls{DM}, \gls{PDM}, \gls{CLLF}, \gls{EDF}, \gls{EPD}, \gls{EDZL} algorithms. 
{To evaluate the performance of the above multi-processor scheduling algorithms, we include the \gls{RANDOM}, which randomly select a maximal set of released transmissions for the schedule of each slot.}
In addition, we include algorithms ALICE \cite{kim2019:alice-autonomous-link-based-cell-scheduling-tsch} and \gls{TASA} \cite{palattella2012:traffic-aware-scheduling-algorithm-reliable-low-power-multihop-ieee-802-15-4e-networks} as baselines, which have also been adapted to our framework\footnote{As shown in \cite{kim2019:alice-autonomous-link-based-cell-scheduling-tsch}, a shorter slot-frame length favors higher performance of ALICE, as it leads to a decreased delay. Therefore, for optimal performance, we choose a slot-frame length of $ 3 $. To implement \gls{TASA}, we find in each slot a maximal matching of no more than $ C $ links from the released transmissions favoring links with longer sender queue. }. Further, we compare the performance of \gls{HS} and \gls{RS}, with and without \gls{OA}.}

\begin{table*} [!t]
\centering
\begin{tabular} { c | c | c | c | c | c | c | c | c | c | c | c} 
	\hline
	\hline
	Algorithm & \bf LLF-RC (Proposed) & \bf RM & \bf DM & \bf PDM & \bf CLLF & \bf EDF & \bf EPD & \bf EDZL & \bf ALICE & \bf TASA & \bf RANDOM \\ 
	\hline
	Implicit deadline & $15.4$  & $7.5$  & $7.5$  & $8.0$  & $396.4$  & $7.3$  &  $7.8$  & $7.7$  & $7.7$  & $7.6$  & $6.6$ \\ \hline
	Restricted deadline & $5.4$ & $3.0$  & $2.9$  & $3.1$  & $43.8$  & $2.9$  & $3.0$  & $3.0$  & $3.1$  & $2.9$  & $2.7$ \\ \hline
\end{tabular}
\caption{Execution time (in ms) of the algorithms using HS.}
\label{tab:execution-time}
\end{table*}

\begin{figure*}[!t] 
\begin{subfigure}{0.48\textwidth}
   	\centering
	\includegraphics[width=\textwidth]{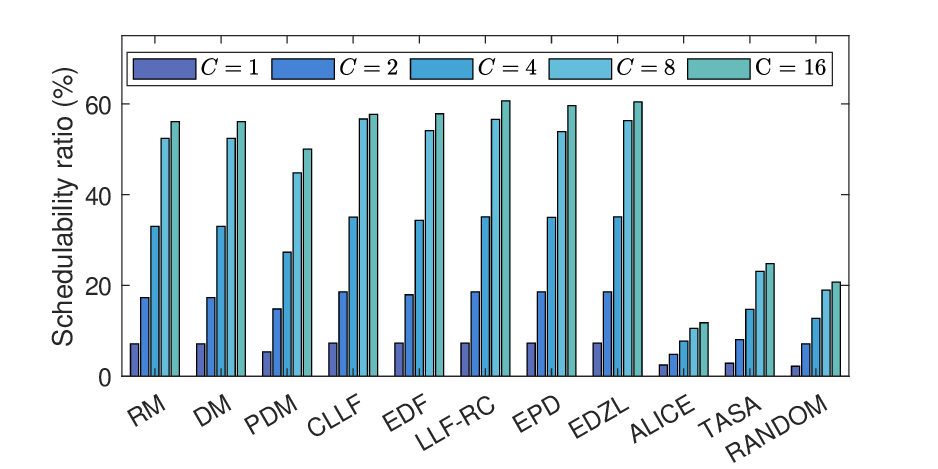}
	\caption{Implicit deadline}
   	\label{fig:schedulability-ratio-algorithms-a}
\end{subfigure}
\hfill
\begin{subfigure}{0.48\textwidth}
	\centering
	\includegraphics[width=\textwidth]{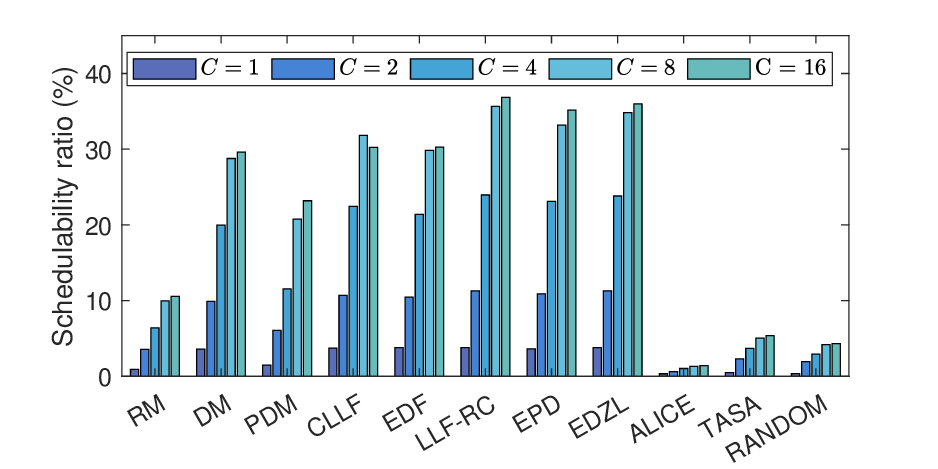}
	\caption{Restricted deadline}
   	\label{fig:schedulability-ratio-algorithms-b}
\end{subfigure}
\caption{Schedulability ratio performance of various algorithms for \emph{2P scheduling}.}
\label{fig:schedulability-ratio-algorithms}
\end{figure*}

\subsection{Reliability of 1P Scheduling and 2P Scheduling} 
To evaluate how much \emph{2P scheduling} increases communication reliability over \emph{1P scheduling}, we generate several topologies and flows as described in \cref{sec:generation-scenarios}. Specifically, we generate the link qualities following uniform distributions. In \cref{fig:comparison-reliability-1p-2p-scheduling}, we evaluate the communication reliability of \emph{1P scheduling} and \emph{2P scheduling} for four cases where the link quality varies in different ranges.

\cref{fig:comparison-reliability-1p-2p-scheduling-02} to \cref{fig:comparison-reliability-1p-2p-scheduling-08} show the empirical \gls{CDF} of the communication reliability of \emph{1P scheduling} and \emph{2P scheduling}. In this scenario, the links qualities are randomly selected from the sets $ [0.2, 1] $, $ [0.4, 1] $, $ [0.6, 1] $, and $ [0.8, 1] $, respectively. We observe that the communication reliability of \emph{2P scheduling} is superior to that of \emph{1P scheduling}. In particular, it exhibits an approximate increment with respect to the median value, between $ 13\% $ (in \cref{fig:comparison-reliability-1p-2p-scheduling-02}) and $ 2\% $ (in \cref{fig:comparison-reliability-1p-2p-scheduling-08}). The performance gap is smaller in \cref{fig:comparison-reliability-1p-2p-scheduling-08} since all link qualities are very high.

\emph{\textsc{Remark:} Via simulations we have shown that the proposed 2P scheduling has higher reliability than the 1P scheduling that is commonly used in the literature, confirming our theoretical analysis in Theorem \ref{thm:reliability}. In the remainder of scenarios, we  focus on 2P scheduling.}


\subsection{Schedulability Ratio of Several Algorithms using HS}

In this section, our goal is to find the best scheduling algorithm without considering optional features, such as \gls{RS} or \gls{OA}. In particular, we investigate the performance of several algorithms in terms of schedulability ratio and execution time.

In \cref{fig:schedulability-ratio-algorithms}, we show the schedulability ratio for various configurations of number of channels and deadline type. 
{We observe that the schedulability ratio generally increases with the number of channels available for any scheduling algorithm, because with more channels, it allows more links to be scheduled simultaneously, and thus it is more likely for the scheduling problem to meet the deadlines.}
The proposed \gls{LLF-RC} algorithm has the highest schedulability ratio in almost all cases, either with an implicit or restricted deadline. Besides, \gls{EDZL} also shows a high schedulability ratio, being only outperformed by \gls{LLF-RC}. In \cref{tab:execution-time}, we show the execution time of all benchmarked algorithms, where we observe that \gls{LLF-RC} needs approximately twice the execution time compared to the fastest algorithm \gls{RANDOM} with either an implicit or constrained deadline. Yet, the execution time of \gls{LLF-RC} with either deadline type is still small enough to allow frequent online rescheduling, which can further improve the schedulability performance.

 \begin{figure*}[!t] 
 	\centering
 	\begin{subfigure}{.48\textwidth}
 	    \centering
 		\includegraphics[width=\textwidth]{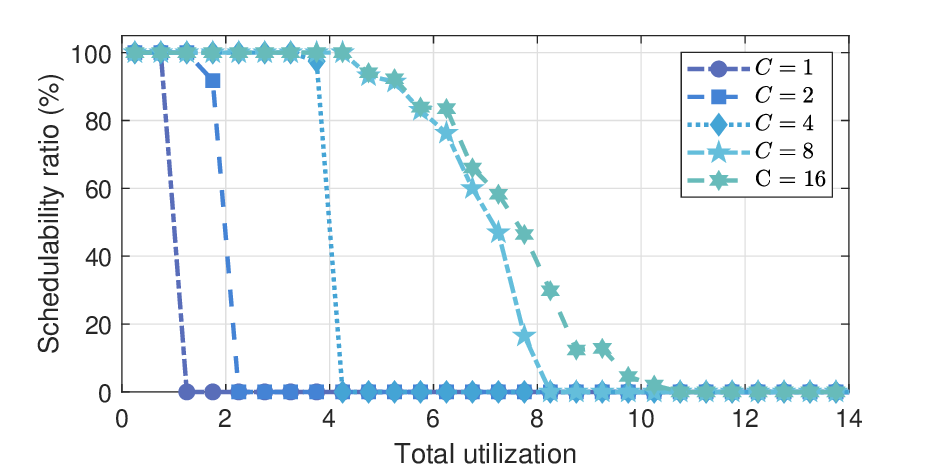}
 		\caption{HS without OA (Implicit deadline)}
 	    \label{fig:schedulability-ratio-vs-channel-utilization-a}
 	    \vspace{2mm}
 	\end{subfigure}
 	\hfill
 	\begin{subfigure}{.48\textwidth}
 	    \centering
 		\includegraphics[width=\textwidth]{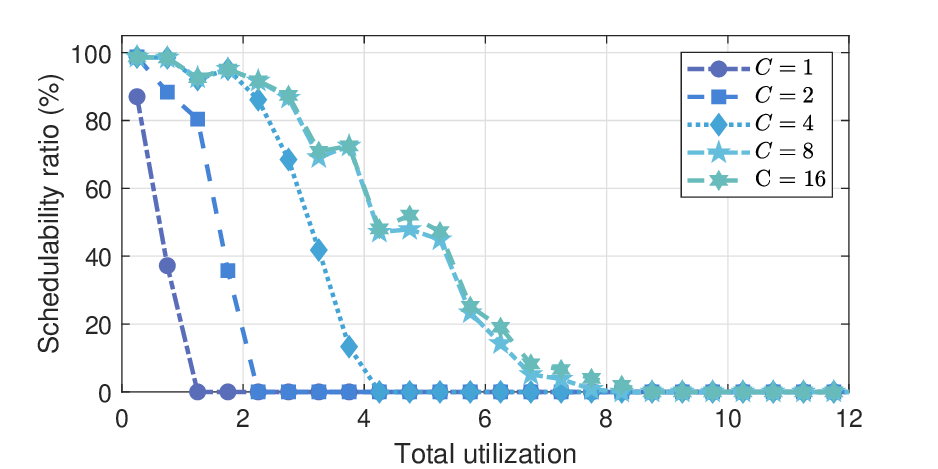}
 		\caption{HS without OA (Restricted deadline)}
 	    \label{fig:schedulability-ratio-vs-channel-utilization-b}
 	    \vspace{2mm}
 	\end{subfigure}
 	\begin{subfigure}{.48\textwidth}
 	   	\centering
 		\includegraphics[width=\textwidth]{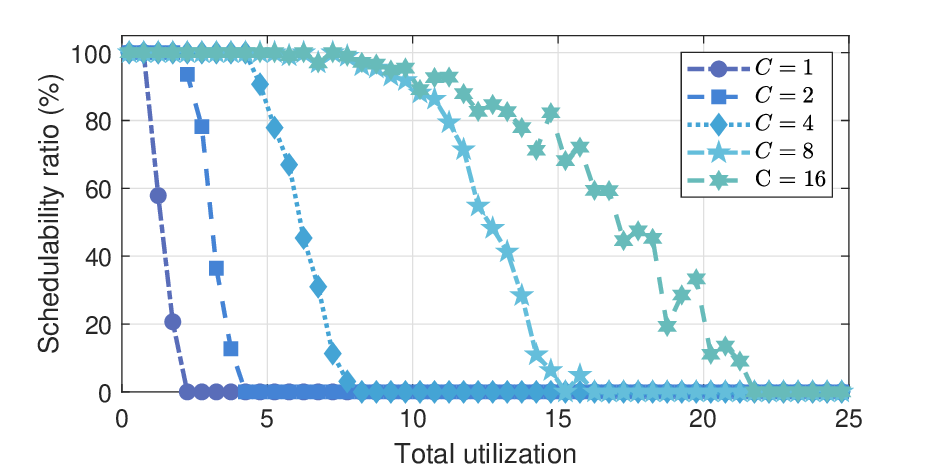}
 		\caption{HS with OA (Implicit deadline)}
 	   	\label{fig:schedulability-ratio-vs-channel-utilization-c}
 	\end{subfigure}
 	\hfill
 	\begin{subfigure}{.48\textwidth}
 		\centering
 		\includegraphics[width=\textwidth]{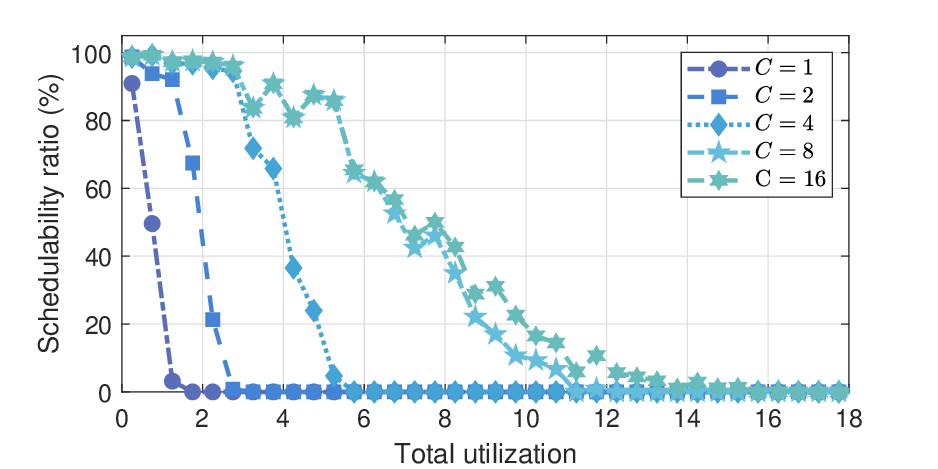}
 		\caption{HS with OA (Restricted deadline)}
 	   	\label{fig:schedulability-ratio-vs-channel-utilization-d}
 	\end{subfigure}
    	\caption{Schedulability ratio of \gls{LLF-RC} using HS with and without OA for various numbers of channels and total utilization. A point $ (x, y) $ on the curves corresponds to the schedulability ratio $ y $ of all scheduling problems which have its total utilization in the range of $ (x-0.25, x+0.25] $. The separation of two neighboring points on a curve is $ 0.5 $ in the $ x $-axis.}
 	\label{fig:schedulability-ratio-vs-channel-utilization}
 \end{figure*}

For 1P scheduling, \cite{saifullah2010:real-time-scheduling-wirelesshart-networks} reported that \gls{CLLF} outperforms \gls{EDF} and \gls{EPD} in terms of schedulability ratio. But for 2P scheduling,  this is not the case, especially for restricted deadlines.
We observe that the schedulability ratio of \gls{CLLF} for restricted deadlines may decrease when the number of channels increases, as shown \cref{fig:schedulability-ratio-algorithms-b} with $ C = 8 $ and $ C = 16 $. Another drawback is that \gls{CLLF} is significantly more time-consuming than the other algorithms, and is between $ 8 $ to $ 26 $ times slower than \gls{LLF-RC}.

For implicit deadlines, fixed-priority scheduling algorithms, such as \gls{RM}, \gls{DM}, and \gls{PDM}, perform slightly worse than the dynamic ones, such as \gls{CLLF}, \gls{EDF}, \gls{EPD}, \gls{EDZL}, and \gls{LLF-RC}. For instance, \gls{RM}, \gls{DM}, and \gls{PDM} are up to $ 4.56\% $, $ 4.56\% $, and $ 10.62\% $ less efficient than \gls{LLF-RC}, respectively, in terms of schedulability ratio. Fixed-priority scheduling algorithms should be avoided for restricted deadlines as their performance is significantly worse than the dynamic ones. For instance, \gls{RM}, \gls{DM}, and \gls{PDM} are up to $ 26.28\% $, $ 7.22\% $, and $ 14.89\% $ less efficient than \gls{LLF-RC}, respectively, in terms of schedulability ratio. 
The various multi-processor scheduling algorithms outperform \gls{RANDOM} by 2 to 11 times, in terms of schedulability ratio, which demonstrates the power of applying ideas drawn from the research on multi-processor scheduling  to \gls{WSAN} scheduling.
Although ALICE, \gls{TASA} and \gls{RANDOM} are competitive in terms of execution time, they are outperformed by the other algorithms in terms of schedulability ratio, showing that these best-effort algorithms are not suitable for real-time scheduling, especially when restricted deadlines are considered. 
\gls{RANDOM} arbitrarily selects a maximal transmission set to be scheduled per slot;
\gls{TASA} optimizes schedule length while ALICE aims to minimize scheduling overhead by computing the schedule locally at every node. None of them strives to meet the deadlines.
ALICE performs the worst among the evaluated algorithms, because it is an autonomous algorithm for which a link is only allowed to be scheduled in slots pre-determined by the source and destination IDs of the link. Thus, ALICE cannot guarantee to schedule a maximal set of transmissions per slot while all the other algorithms can.

\emph{\textsc{Remark:} Given our findings, which highlight the superior schedulability ratio of the proposed \gls{LLF-RC} algorithm, which we identify as the most relevant criterion. Our subsequent analysis is focused on conducting a comprehensive evaluation of \gls{LLF-RC} in various scenarios.}

\subsection{Performance of LLF-RC using HS with and without OA} 

We investigate the performance of \gls{LLF-RC} using \gls{HS} with and without \gls{OA}, in terms of schedulability ratio, execution time, and maximum queue length, under various parameter settings\footnote{The results for \gls{HS} without \gls{OA} are obtained using \cref{alg:framework} whereas the results for \gls{HS} with \gls{OA} are obtained using both \cref{alg:framework} and \cref{alg:opportunistic-aggregation} in tandem. For the latter case, \texttt{line 6} of \cref{alg:framework} is replaced by \cref{alg:opportunistic-aggregation}.}.

%
%

\textbf{Impact of number of channels and total utilization.} \cref{fig:schedulability-ratio-vs-channel-utilization} shows the schedulability ratio of \gls{LLF-RC} using \gls{HS} without \gls{OA} (in \cref{fig:schedulability-ratio-vs-channel-utilization-a} and \cref{fig:schedulability-ratio-vs-channel-utilization-b}) and with \gls{OA} (in \cref{fig:schedulability-ratio-vs-channel-utilization-c} and \cref{fig:schedulability-ratio-vs-channel-utilization-d}), considering different numbers of channels and total utilization. Throughout \cref{fig:schedulability-ratio-vs-channel-utilization-a} to \cref{fig:schedulability-ratio-vs-channel-utilization-d}, we observe that the schedulability ratio decreases monotonically when the total utilization increases, for a given number of channels. This occurs because busier networks are less schedulable. In addition, the use of implicit deadlines results in a larger feasible utilization region than with restricted deadlines. The reason is that scheduling with deadlines smaller than periods is more likely to be infeasible than scheduling with equal deadlines and periods, especially in busier networks, which have higher utilization. 

\begin{figure*}[!t] 
	\centering
	\begin{subfigure}{0.48\textwidth}
	   	\centering
		\includegraphics[width=\textwidth]{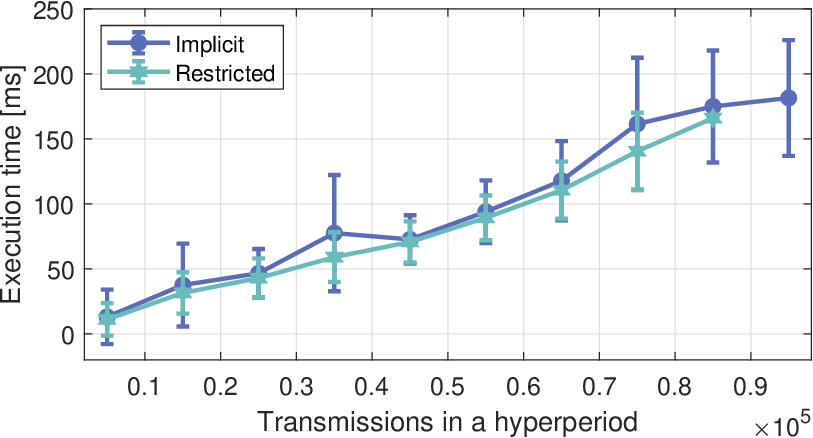}
		\caption{HS without OA}
	   	\label{fig:execution-time-a}
	\end{subfigure}
	\hfill
	\begin{subfigure}{0.48\textwidth}
		\centering
		\includegraphics[width=\textwidth]{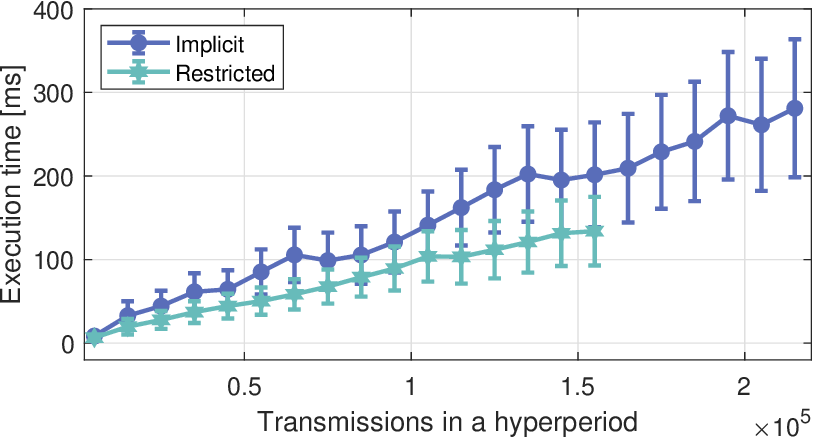}
		\caption{HS with OA}
	   	\label{fig:execution-time-b}
	\end{subfigure}
   	\caption{Execution time of LLF-RC using HS with and without OA. A point $(x, y)$ on the curves corresponds to the mean execution time ($y$) of all schedulable problems with number of transmissions in $ (x-0.5 \cdot 10^4, x+0.5 \cdot 10^4) $, $ x = \left\lbrace 0.5, 1.5, \dots \right\rbrace \cdot 10^4 $. The error bars denote $\pm$ one standard deviation.
   	}
   	\label{fig:execution-time}
\end{figure*}

\begin{figure*}[!t] 
	\centering
	\begin{subfigure}{0.48\textwidth}
	   	\centering
		\includegraphics[width=\textwidth]{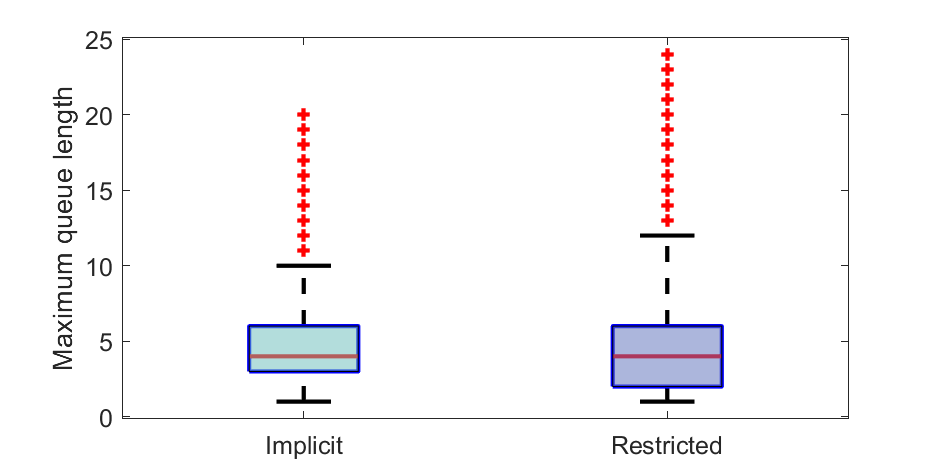}
		\caption{HS without OA}
	   	\label{fig:maximum-queue-length-a}
	\end{subfigure}
	\hfill
	\begin{subfigure}{0.48\textwidth}
		\centering
		\includegraphics[width=\textwidth]{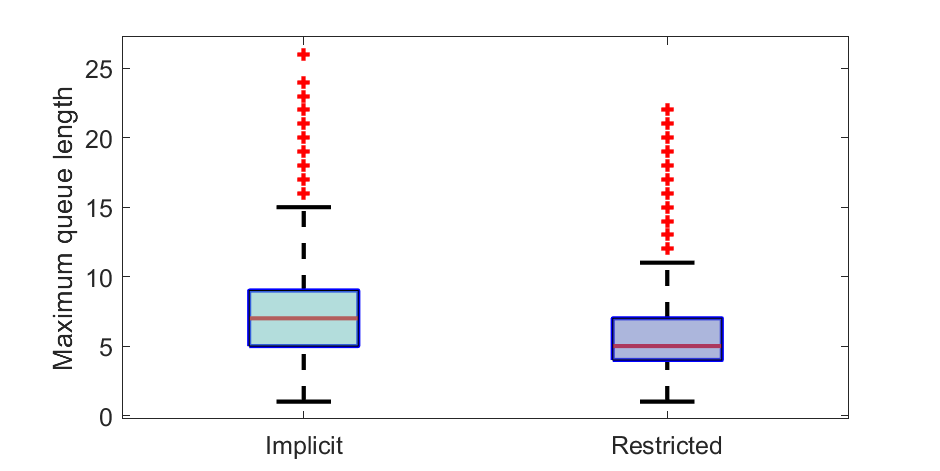}
		\caption{HS with OA}
	   	\label{fig:maximum-queue-length-b}
	\end{subfigure}
   	\caption{Maximum queue length of \gls{LLF-RC} using HS with and without OA. Box plots for various values of maximum queue length.}
   	\label{fig:maximum-queue-length}
\end{figure*}

From \cref{fig:schedulability-ratio-vs-channel-utilization-a} to \cref{fig:schedulability-ratio-vs-channel-utilization-d}, we observe that using \gls{HS} with \gls{OA} significantly improves the schedulability ratio for both implicit and restricted deadlines, compared to using \gls{HS} without \gls{OA}. In particular, the relative improvement obtained by \gls{OA} is greater when implicit deadlines are considered. The reason is that implicit deadlines are less restrictive, which gives more possibilities for aggregation. For instance, considering $ C = 8 $ and implicit deadlines, we observe that the use of \gls{OA} contributes to increasing the schedulability ratio from $ 16\% $ to $ 97 \%$. For restricted deadlines, the use of \gls{OA} improves the schedulability ratio from $ 3\% $ to $ 46 \%$. Thus, using \gls{HS} with \gls{OA} is very effective in combating conflicts, pushing the maximum total utilization far beyond the number of channels, especially for implicit deadlines.  

\gls{OA}'s aggregation rate, i.e. the percentage of combined packets among all packets, is on average $ 23\% $ and $ 17\% $, for implicit and restricted deadlines, respectively. Thus, we note that with a moderate aggregation of packets, we can increase the schedulability ratio significantly.

\textbf{Execution time.} In \cref{fig:execution-time}, we show the execution time of \gls{LLF-RC} using \gls{HS} with and without \gls{OA}. We observe that the execution time increases almost linearly with the number of transmissions in both cases. Moreover, for the same number of transmissions, implicit deadline requires a relatively larger execution time. 

It is worth noting that the adoption of \gls{OA} results in a reduction in execution time ranging from $ 14\% $ to $ 29\% $ compared to the case without \gls{OA}. The shorter execution time achieved when using \gls{OA} is attributed to a higher throughput, which is consequence of a more compact scheduling table with fewer slots containing non-empty schedules. Still, it is important to recognize that \gls{OA} requires more time to schedule a single slot due to the higher complexity of \cref{alg:opportunistic-aggregation}.


\textbf{Maximum queue length.} As sensor nodes usually have very limited memory resources, it is vital to investigate the memory consumption during network operation, for which we use the maximum queue length metric. When evaluating this metric, we only take into account the motes, without considering the gateways, since the latter usually have sufficient storage capacity.

From our evaluation, we found that the maximum queue length increases linearly with the number of flows and total utilization. Also, the maximum queue length increases almost linearly with the number of channels. This occurs because more channels allows for more simultaneous transmissions, leading to a more dynamic behavior of the maximum queue length. This leads to more flows and more transmissions ready to be scheduled per slot, and consequently to a higher number of potential contentions per link. Therefore, packets are more likely to be queued at a node. From \cref{fig:maximum-queue-length-a} and \cref{fig:maximum-queue-length-b}, we observe that the maximum queue length is higher when \gls{OA} is used, since the higher throughput provided by \gls{OA} causes more packets to be buffered at intermediate nodes. 

In \cref{fig:maximum-queue-length-a}, we observe that the median value of the maximum queue length is approximately $ 5 $ for implicit and restricted deadlines. However, the worst-case occurs for restricted deadlines, where up to $ 24 $ packets are queued. Since a packet has a payload size of at most $ 127 $ bytes, this requires less than $ 3 $ KBytes, which can fit into the RAM of a commercial sensor node. In \cref{fig:maximum-queue-length-b}, we show a similar case, where the maximum queue length has a median value of less than $ 6 $ for implicit and restricted deadlines. The worst-case occurs for implicit deadlines, where up to $ 26 $ packets are queued. However, this is also small enough to fit into the RAM of a sensor node. 

We observe that the penalty in maximum queue length incurred by the use of \gls{OA} is affordable in exchange for a significant improvement in schedulability ratio and reduction of execution time.

\emph{\textsc{Remark:} {Simulation results confirm that the proposed \gls{OA} mechanism increases the schedulability ratio for scheduling problems, and shortens the execution time of scheduling algorithm, with only slight increase in buffer size at nodes.}}

\begin{figure*}[!t] 
	\centering
	\begin{subfigure}{0.48\textwidth}
	   	\centering
		\includegraphics[width=\textwidth]{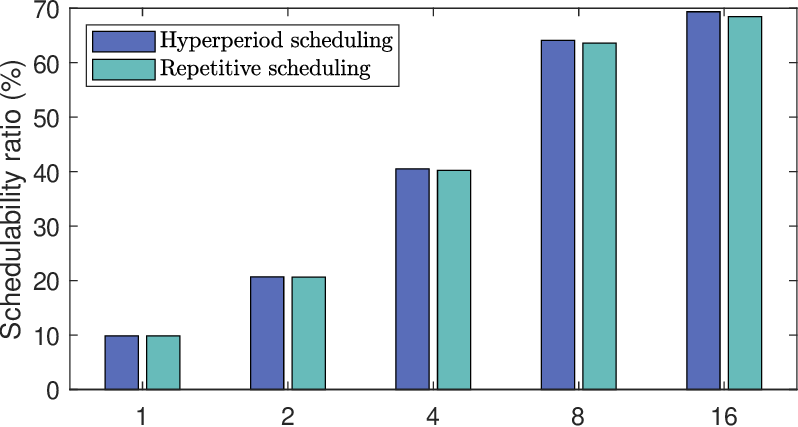}
		\caption{Without \gls{OA} (Implicit deadline)}
	   	\label{fig:repetitive-vs-hyperperiod-scheduling-a}
	\end{subfigure}
	\hfill
	\begin{subfigure}{0.48\textwidth}
		\centering
		\includegraphics[width=\textwidth]{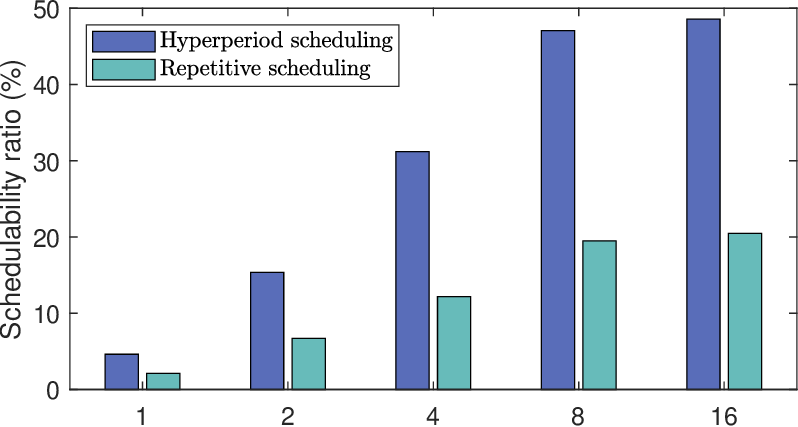}
		\caption{Without \gls{OA} (Restricted deadline)}
	   	\label{fig:repetitive-vs-hyperperiod-scheduling-b}
	\end{subfigure}
	\begin{subfigure}{0.48\textwidth}
	   	\centering
		\includegraphics[width=\textwidth]{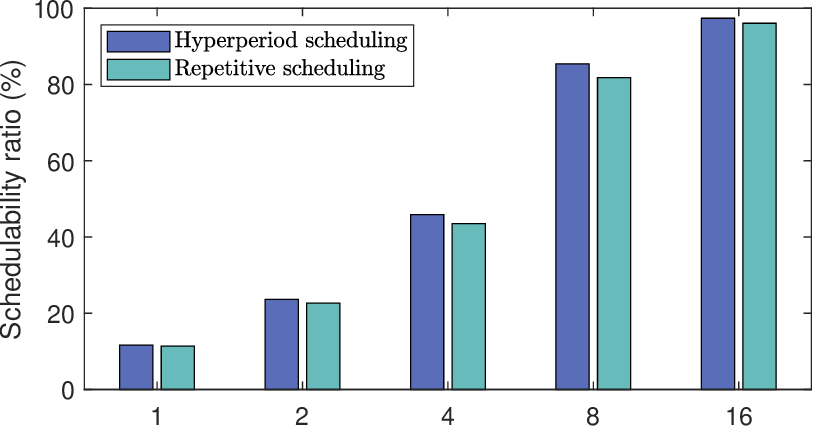}
		\caption{With \gls{OA} (Implicit deadline)}
	   	\label{fig:repetitive-vs-hyperperiod-scheduling-c}
	\end{subfigure}
	\hfill
	\begin{subfigure}{0.48\textwidth}
		\centering
		\includegraphics[width=\textwidth]{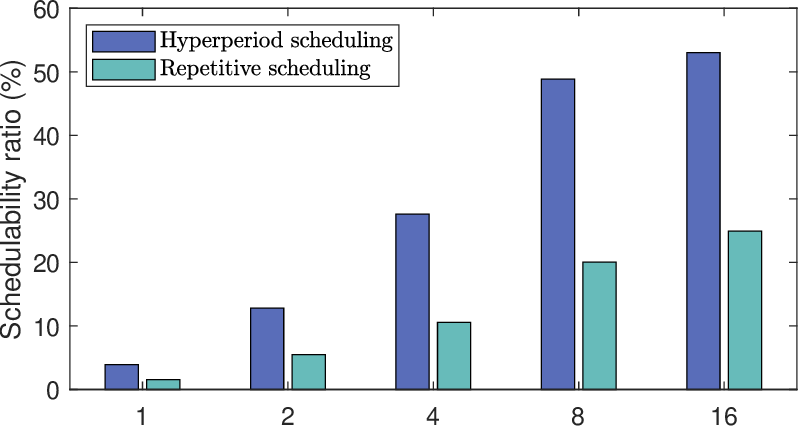}
		\caption{With \gls{OA} (Restricted deadline)}
	   	\label{fig:repetitive-vs-hyperperiod-scheduling-d}
	\end{subfigure}
   	\caption{Schedulability ratio of LLF-RC using RS or HS.}
   	\label{fig:repetitive-vs-hyperperiod-scheduling}
\end{figure*}

\subsection{Performance of LLF-RC using HS or RS}  \label{sec:perf-llc-rc-hs-rs}

We evaluate the performance of \gls{LLF-RC} using \gls{RS} or \gls{HS}, in terms of schedulability ratio, scheduling table size, and execution time.

\textbf{Schedulability ratio.} In \cref{fig:repetitive-vs-hyperperiod-scheduling}, we show the schedulability ratio of \gls{RS} and \gls{HS}, respectively. For implicit deadlines, we observe that \gls{RS} decreases the schedulability ratio only slightly compared to \gls{HS}. For restricted deadlines, however, \gls{RS} causes a large decrement in the schedulability ratio, approximately $ 50\% $. In practice, most control systems operate under implicit deadline requirements, thus making \gls{RS} still highly attractive due to its reduced scheduling table size and execution time, which are discussed next.

\textbf{Scheduling table size.} \cref{fig:cost-savings} shows the number of used entries in the scheduling table when \gls{RS} or \gls{HS} is employed. Because the communication and storage costs are proportional to the number of used entries, we observe a tremendous cost reduction with \gls{RS}. We realize from \cref{fig:cost-savings-a} and \cref{fig:cost-savings-b} that the maximum number of used entries are $ 1284 $ and $ 150396 $, for \gls{RS} and \gls{HS}, respectively. Each entry in the scheduling table contains the fields: slot ID ($ 2 $ bytes), channel ID ($ 4 $ bits), sender ID ($ 1 $ byte), receiver ID ($ 1 $ byte) and flow ID ($ 6 $ bits). In total, the size of each entry is about $ 5 $ bytes. Thus, the maximum scheduling table size for \gls{RS} and \gls{HS} are $ 51.4 $ kbits and $ 6015.8 $ kbits, respectively. As shown by \cite{yuan2015:ripple-high-throughput-reliable-energy-efficient-network-flooding-wireless-sensor-networks}, the Ripple protocol can achieve a flooding throughput of about $ 90 $ kbit/s. Therefore, the scheduling table for \gls{RS} can be downloaded in less that $ 0.6 $s, allowing frequent schedule refresh. On the other hand, with \gls{HS}, more than $ 66 $s would be required, making it less practical. In addition, since the scheduling table has to be transmitted and stored at the wireless nodes, there are significant savings in communication and storage costs, which is especially critical for wireless nodes with limited resources, such as sensors and relays.

\begin{figure*}[!t] 
	\centering
	\begin{subfigure}{1\columnwidth}
	   	\centering
		\includegraphics[width=\textwidth]{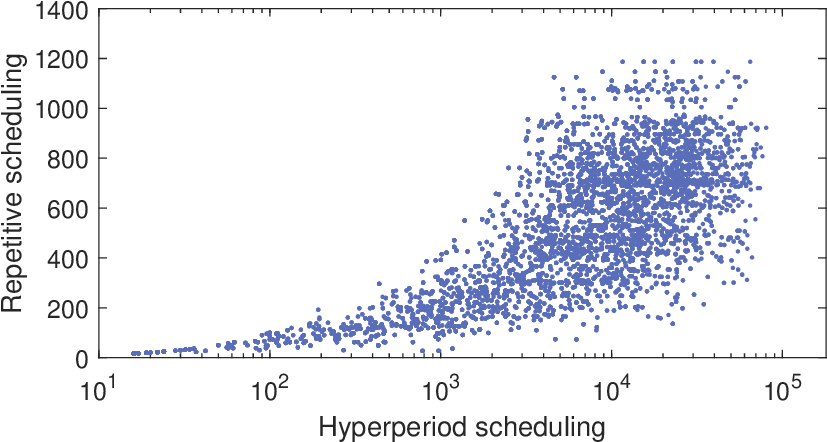}
		\caption{Without OA (Implicit deadline)}
	   	\label{fig:cost-savings-a}
	\end{subfigure}
	\hfill
	\begin{subfigure}{1\columnwidth}
		\centering
		\includegraphics[width=\textwidth]{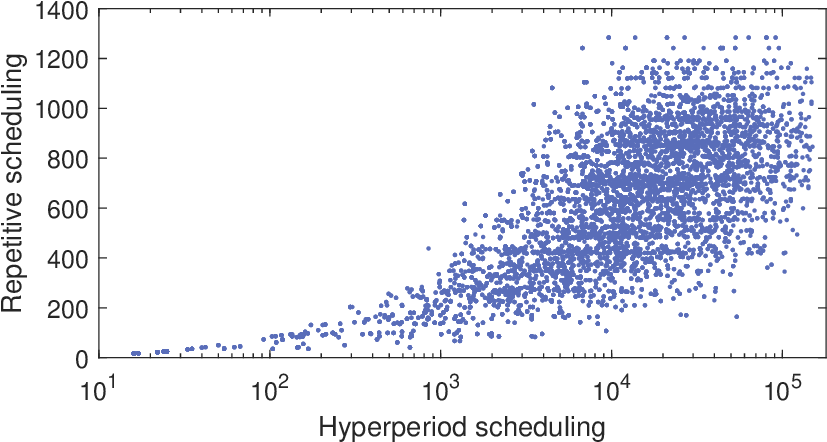}
		\caption{With OA (Implicit deadline)}
	   	\label{fig:cost-savings-b}
	\end{subfigure}
   	\caption{Number of used entries in the scheduling table using LLF-RC with RS or HS.}
   	\label{fig:cost-savings}
\end{figure*}

\textbf{Execution time.} In \cref{fig:execution-time-repetitive-hyperperiod-scheduling}, we compare the execution time of \gls{HS} and \gls{RS}. We observe that when \gls{HS} is employed, the execution time grows linearly for an increasing number of transmissions. In contrast, the execution time remains almost constant when \gls{RS} is employed, showing its potential for real-time control applications. 

{For a network of 100 nodes and up to 50 flows, the execution time of scheduling algorithm \gls{RS} is in most cases below 10ms. By applying \gls{OA}, the execution time can be further reduced by half. Thus, the scheduling scheme \gls{LLF-RC} + \gls{RS} + \gls{OA} proposed in this paper is practical for medium-size industrial networks.}

\emph{\textsc{Remark:} 
{Simulation results confirm that for scheduling problems of implicit deadline,
the proposed \gls{RS} mechanism reduces the scheduling table size and the execution time of scheduling algorithm tremendously, with only slight decrease in schedulability ratio.}}

\begin{figure*}[!t] 
	\centering
	\begin{subfigure}{0.48\textwidth}
	   	\centering
		\includegraphics[width=\textwidth]{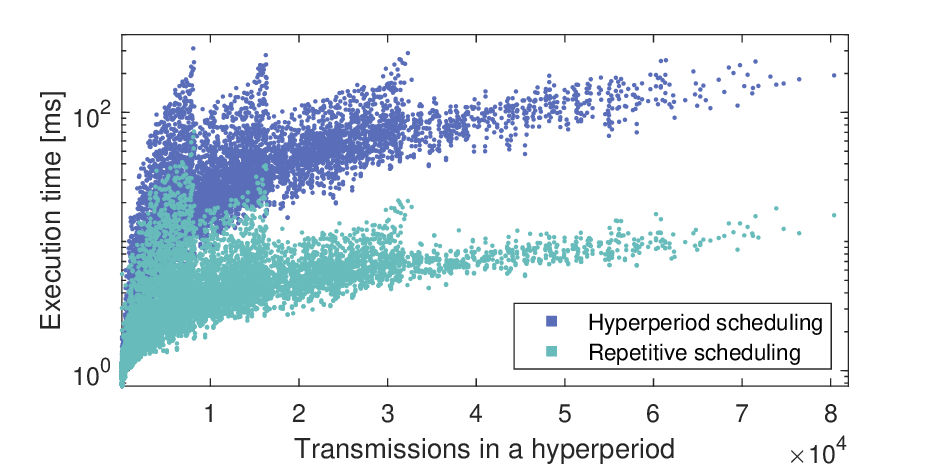}
		\caption{Without OA (Implicit deadline)}
	   	\label{fig:execution-time-repetitive-hyperperiod-scheduling-a}
	\end{subfigure}
	\hfill
	\begin{subfigure}{0.48\textwidth}
		\centering
		\includegraphics[width=\textwidth]{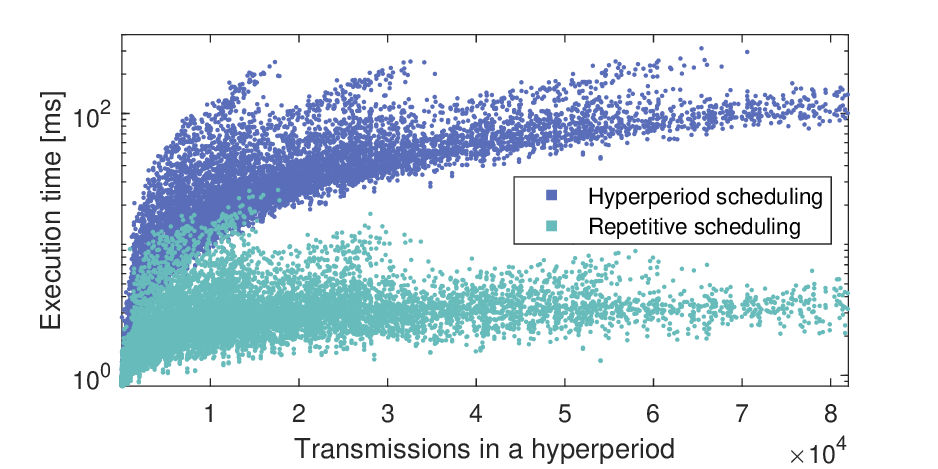}
		\caption{With OA (Implicit deadline)}
	   	\label{fig:execution-time-repetitive-hyperperiod-scheduling-b}
	\end{subfigure}
	\caption{Execution time of LLF-RC using \gls{HS} or \gls{RS}.}
   	\label{fig:execution-time-repetitive-hyperperiod-scheduling}
\end{figure*}

%% file: sections/extension-large-network.tex
\section{Extension to Large-scale Real-world Industrial Networks} \label{sec:extension-large-net}


So far, the algorithms proposed in this paper are only suitable for small to medium size industrial networks, because frequency (communication channel) reuse is disallowed. Therefore, the network capacity is limited, and will not scale with network size (number of nodes). To extend the proposed algorithms to large-scale real-world industrial networks, we have to leverage frequency reuse in the manner of cellular networks.

\begin{figure}[!h] 
   \centering
   \includegraphics[width=5cm]{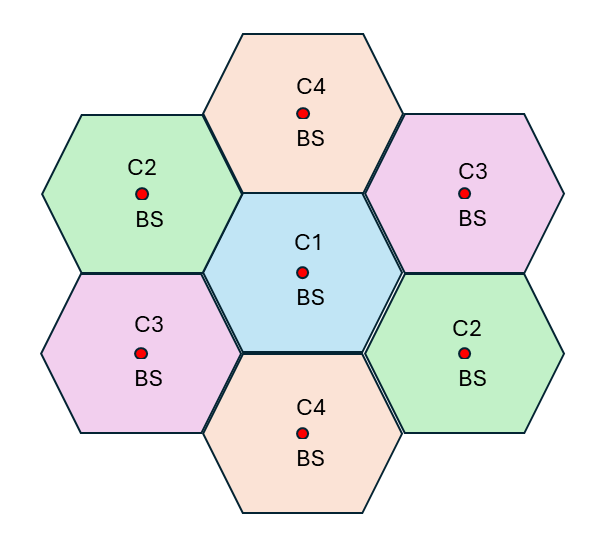}
   \caption{Illustration of a large-scale network divided into cells. We have 4 types of cells (C1, C2, C3, C4). Each cell uses 4 communication channels. }
   \label{fig:large-network}
\end{figure}

As illustrated in Fig.~\ref{fig:large-network}, a large-scale network can be divided into many cells. Each cell has a base station, and all nodes in the cell can only use a certain set of communication channels. Two neighboring cells must use different sets of communication channels to reduce mutual interference. The base stations are connected by a wired backbone network. 

\textbf{Routing.} A sensing packet must be routed from a sensor to the base station of the cell using  communication channels that are allowed by the cell. An actuating packet must be routed to an actuator from the base station of the cell using communication channels that are allowed by the cell.
Hence, the hop distance from a sensor to an actuator is in the worst case twice the cell diameter, irrespective of the network size. Cell diameter is defined at the largest hop distance of two nodes in a same cell. So the flow reliability and  the flow delay are independent of network size.
The same multi-path routing model in Section~\ref{sec:routing-model} can be applied for large-scale networks.

\textbf{Scheduling algorithms.} The scheduling algorithms proposed in Section~\ref{sec:proposed-scheduling-framework} can be easily extended to support large-scale networks. The only needed modification is that links in different cells can be scheduled in a same slot. In this way, the network throughput is linear to the number of cells, and thus, it is proportional to the network size.

\textbf{Execution time.} The execution time of scheduling algorithms should be linear to network size and number of flows if a central unit performs the scheduling in a serial way. Yet, the execution time can be reduced significantly since the link scheduling per slot can be parallelized. This is due to the fact that any two links from different cells can always be scheduled in a same slot because of the cellular division of a large network. Thus, the link scheduling per slot can be done in parallel for each cell.

\color{black}

%% file: sections/conclusions.tex

\section{Conclusion} \label{sec:conclusion}

In this paper, we investigated centralized \gls{TDMA} scheduling of \glspl{WSAN} for multi-rate periodic control systems in Industry 4.0. We proposed \emph{2P scheduling}, a new scheduling framework with higher reliability than conventional \gls{WSAN} scheduling. 
Drawing from multiprocessor-scheduling, we further proposed \gls{LLF-RC}, an efficient scheduling algorithm that prioritizes flows based on link conflicts.
Simulations showed \gls{LLF-RC} has the highest schedulability ratio among benchmarked algorithms.
It has low execution time (in most cases below 10ms for a network of 100 nodes and up to 50 flows) and low maximum queue length, making it practical for real-world systems.
We also proposed \gls{OA} and \gls{RS}, two new scheduling strategies that can be easily integrated with any scheduling algorithm.
\gls{OA} increases schedulability ratio by up to 97\% in our simulation, and pushes total utilization beyond the number of channels in the system, suitable for congested networks.
\gls{RS} reduces the maximum scheduling table size and the maximum execution time by 99\% and 92\% respectively in our simulation, leading to high energy savings at wireless nodes.
We evaluated these techniques in various scenarios, demonstrating that \gls{LLF-RC}, in combination with \gls{OA} and \gls{RS}, provides an effective and efficient solution for real-time TDMA scheduling in industrial \glspl{WSAN}.
Finally, we extended the proposed scheduling algorithms to support large-scale networks by exploiting cellular frequency reuse.
There, network throughput increases linearly with network size, flow reliabilities and delays remain independent of network size, and the scheduling algorithms can be parallelized.

%
%
%
%

%% file: sections/appendices.tex

\begin{appendices}

\renewcommand{\thesectiondis}[2]{\Alph{section}:}

\setcounter{equation}{0}
\renewcommand{\theequation}{A.\arabic{equation}}
\section{Reliability of 1P and 2P Scheduling} \label{app:thm-1}

Let us introduce $ l_i \geq 0 $ and $ m_i \geq 0 $ such that $ l_i = (1 - r_i^\mathrm{sc}) $ and $ m_i = (1 - r_i^\mathrm{ca}) $. Therefore, $ r^\mathrm{1P} $ and $ r^\mathrm{2P} $ are equivalently expressed as
\begin{align}
	r^\mathrm{1P} & = 1 - \prod_{i = 0}^{n-1} (1 - (1 - l_i)(1 - m_i)) \nonumber \\
				  & = 1 - \prod_{i = 0}^{n-1} (l_i + m_i - l_i m_i) \nonumber
\end{align}
\begin{align}
	r^\mathrm{2P} & = \left( 1 - \prod_{i = 0}^{n - 1} l_i \right) \left( 1 - \prod_{i=0}^{n-1} m_i \right) \nonumber \\
				  & = 1 - \prod_{i = 0}^{n - 1} l_i - \prod_{i = 0}^{n - 1} m_i + \prod_{i = 0}^{n - 1} l_i m_i \nonumber
\end{align}

Let us introduce $ c^\mathrm{1P} \leq 0 $ and $ c^\mathrm{2P} \leq 0 $, defined below,
\begin{align*}
	c^\mathrm{1P} & = - \prod_{i = 0}^{n - 1} (l_i + m_i - l_i m_i),  \\
	c^\mathrm{2P} & = - \prod_{i = 0}^{n - 1} l_i - \prod_{i = 0}^{n - 1} m_i + \prod_{i = 0}^{n - 1} l_i m_i,  
\end{align*}
yielding $ r^\mathrm{1P} = 1 + c^\mathrm{1P} $ and $ r^\mathrm{2P} = 1 + c^\mathrm{2P} $. Hence, proving that $ r^\mathrm{2P} \geq r^\mathrm{1P} $ is equivalent to proving that $ c^\mathrm{1P} \leq c^\mathrm{2P} $.

For notation simplicity, let us assume that $ k = n- 1 $. Thus, $ c^\mathrm{1P} $ can be expressed as follows,
\begin{align*}
	c^\mathrm{1P} = & - \prod_{i = 0}^{k} (l_i + m_i - l_i m_i), 
	\\
				  = & - \prod_{i = 0}^{k - 1} (l_i + m_i - l_i m_i) (l_k + m_k - l_k m_k), 
	\\
				  = & \left( - \prod_{i = 0}^{k - 1} l_i - \prod_{i = 0}^{k - 1} m_i + \prod_{i = 0}^{k - 1} l_i m_i \right) \cdot \left( l_k + m_k - l_k m_k \right), 
	\\
	 			  = & - \prod_{i=0}^{k} l_i - \prod_{i=0}^{k} m_i - \prod_{i=0}^{k} l_i m_i - \prod_{i=0}^{k-1} l_i m_k + \prod_{i=0}^{k-1} l_i l_k m_k
	\\
				  & - \prod_{i=0}^{k-1} m_i l_k + \prod_{i=0}^{k-1} m_i l_k m_k + \prod_{i=0}^{k-1} l_i m_i l_k + \prod_{i=0}^{k-1} l_i m_i m_k, 
	\\
	 			  = & - \prod_{i=0}^{k} l_i - \prod_{i=0}^{k} m_i + \left( \prod_{i=0}^{k} l_i m_i - 2 \prod_{i=0}^{k} l_i m_i \right)
	\\
				  & - \prod_{i=0}^{k-1} l_i m_k + \prod_{i=0}^{k-1} l_i l_k m_k - \prod_{i=0}^{k-1} m_i l_k 
	\\			  
				  & + \prod_{i=0}^{k-1} m_i l_k m_k + \prod_{i=0}^{k-1} l_i m_i l_k + \prod_{i=0}^{k-1} l_i m_i m_k, 
	\\
				 = & - \prod_{i=0}^{k} l_i - \prod_{i=0}^{k} m_i + \prod_{i=0}^{k} l_i m_i 
	\\
				 & - \prod_{i=0}^{k-1} l_i m_i (l_k m_k - m_k) - \prod_{i=0}^{k-1} l_i m_i (l_k m_k - l_k) 
	\\
				 & - \prod_{i=0}^{k-1} l_i m_k (1 - l_k) - \prod_{i=0}^{k-1} m_i l_k (1 - m_k), 
	\\
				= & - \prod_{i=0}^{k} l_i - \prod_{i=0}^{k} m_i + \prod_{i=0}^{k} l_i m_i 
	\\
				 & - \prod_{i=0}^{k-1} l_i m_k (1- l_k) (1 - \prod_{i=0}^{k-1} m_i)  
	\\
				 & - \prod_{i=0}^{k-1} m_i l_k (1- m_k) (1 - \prod_{i=0}^{k-1} l_i). 
\end{align*}

From the above equality, we obtain that $ c^\mathrm{1P} = c^\mathrm{2P} + c' + c'' $, where $ c' < 0 $ and $ c'' < 0 $ are defined as $ c' = \prod_{i=0}^{k-1} l_i m_k (1- l_k) (1 - \prod_{i=0}^{k-1} m_i) $ and $ c'' = \prod_{i=0}^{k-1} m_i l_k (1- m_k) (1 - \prod_{i=0}^{k-1} l_i) $. Thus, we conclude that $ c^\mathrm{2P} \geq c^\mathrm{1P} $, and consequently we claim that  $ r^\mathrm{2P} \geq r^\mathrm{1P} $, thus completing the proof.

\setcounter{equation}{0}
\renewcommand{\theequation}{B.\arabic{equation}}
\section{Latency of 1P and 2P Scheduling} \label{app:thm-2}

Assuming that $\mathrm{hops}(\pi_i)$ denotes the hop count of a given path $ \pi_i $, the minimum flow delay of \emph{2P scheduling} is $ d^\mathrm{2P} = \underset{i}{\mathrm{max}}\{\mathrm{hops}(\pi_i^\mathrm{sc})\} + \underset{i}{\mathrm{min}}\{\mathrm{hops}(\pi_i^\mathrm{ca})\} $, which is equal to the sum hop count of the longest sc-path and the shortest ca-path. We need to make sure that the sensor packet can be delivered to the controller even via the longest sc-path when \emph{2P scheduling} is employed. On the other hand, the minimum flow delay of \emph{1P scheduling} is $ d^\mathrm{1P} = \underset{i}{\mathrm{min}}\{\mathrm{hops}(\pi_i^\mathrm{sc}) + \mathrm{hops}(\pi_i^\mathrm{ca})\} $, which is equal to the hop count of the shortest sa-path. 


To prove that $ d^\mathrm{2P} \geq d^\mathrm{1P} $, we derive a first inequality, which we refer to as $ \mathcal{I}_1 $. Let $ \left\lbrace x_i \right\rbrace_{i=0}^{n-1}, \left\lbrace y_i \right\rbrace_{i=0}^{n-1} \in \mathbb{Z}^{+} \cup \left\lbrace 0 \right\rbrace $ be nonnegative integers, since the hop count has the same nature. Note that $ \underset{i}{\mathrm{min}}\{ x_i \} \leq  x_i $ and $ \underset{i}{\mathrm{min}}\{ y_i \} \leq y_i $ hold true for all $ i $. Using these inequalities, we obtain that $ \underset{i}{\mathrm{min}}\{ x_i \} + \underset{i}{\mathrm{min}}\{ y_i \} \leq x_i + y_i $, which also holds for all $ i $, including the index $ i' $ that attains $ x_{i'} + y_{i'} = \underset{i}{\mathrm{min}} \{ x_{i'} + y_{i'} \} $. Therefore, we have obtained our first inequality $ \mathcal{I}_1: \underset{i}{\mathrm{min}}\{ x_i \} + \underset{i}{\mathrm{min}}\{ y_i \} \leq \underset{i}{\mathrm{min}}\{ x_i + y_i \} $. 

Now, we derive a second inequality, which we refer to as $ \mathcal{I}_2 $. Let $ \left\lbrace w_i \right\rbrace_{i=0}^{n-1}, \left\lbrace z_i \right\rbrace_{i=0}^{n-1} \in \mathbb{Z}^{+} \cup \left\lbrace 0 \right\rbrace $, such that $ x_i = w_i + z_i $. Replacing this equivalence in $ \mathcal{I}_1 $ leads to $ \underset{i}{\mathrm{min}}\{ w_i + z_i \} + \underset{i}{\mathrm{min}}\{ y_i \} \leq \underset{i}{\mathrm{min}}\{ w_i + z_i + y_i \} $. Note than an upper bound for $ w_i $ is $ \underset{i}{\mathrm{max}}\{ w_i \} $, which we can use to obtain a looser inequality  $ \underset{i}{\mathrm{min}}\{ w_i + z_i \} + \underset{i}{\mathrm{min}}\{ y_i \} \leq \underset{i}{\mathrm{min}}\{ \underset{i}{\mathrm{max}}\{ w_i \} + z_i + y_i \} $. Because, $ \underset{i}{\mathrm{max}}\{ x_i \} $ is constant for all $ i $, the inequality can be equivalently expressed as $ \underset{i}{\mathrm{min}}\{ w_i + z_i \} + \underset{i}{\mathrm{min}}\{ y_i \} \leq \underset{i}{\mathrm{max}}\{ w_i \} + \underset{i}{\mathrm{min}}\{ z_i + y_i \} $. By setting $ y_i = 0 $, we obtain our second inequality $ \mathcal{I}_2: \underset{i}{\mathrm{min}}\{ w_i + z_i \} \leq \underset{i}{\mathrm{max}}\{ w_i \} + \underset{i}{\mathrm{min}}\{ z_i \} $. Finally, with the changes of variables, $ w_i = \mathrm{hops}(\pi_i^\mathrm{sc}) $ and $ z_i = \mathrm{hops}(\pi_i^\mathrm{ca}) $, we obtain that $ d^\mathrm{2P} \geq d^\mathrm{1P} $, which completes the proof.

\end{appendices}